\documentclass[letterpaper, 10 pt, conference]{ieeeconf}  

\IEEEoverridecommandlockouts                              

\usepackage{cite}
\usepackage{amsmath,amssymb,amsfonts}
\usepackage{algorithmic}
\usepackage{graphicx}
\usepackage{subcaption}
\usepackage{pgfplots}
\usepackage{tikz}
\pgfplotsset{compat=newest} 

\newtheorem{problem}{Problem}

\newtheorem{prop}{Proposition}
\newtheorem{Oproblem}{Optimization Problem}

\title{\LARGE \bf Multi-Stage Sparse Resource Allocation for Control of Spreading Processes over Networks}

\author{Vera L. J. Somers and Ian R. Manchester
\thanks{The authors are with the Australian Centre for Field Robotics (ACFR), School of Aerospace, Mechanical and Mechatronic Engineering,
        University of Sydney, NSW 2006, Australia
        {\tt\small \{v.somers, i.manchester\}@acfr.usyd.edu.au}}%
}

\begin{document}

\maketitle
\thispagestyle{empty}
\pagestyle{empty}

\begin{abstract}
In this paper we propose a method for sparse dynamic allocation of resources to bound the risk of spreading processes, such as epidemics and wildfires, using convex optimization and dynamic programming techniques. Here, risk is defined as the risk of an undetected outbreak, i.e. the product of the probability of an outbreak occurring over a time interval and the future impact of that outbreak, and we can allocate budgeted resources each time step to bound and minimize the risk. Our method in particular provides sparsity of resources, which is important due to the large network structures involved with spreading processes and has advantages when resources can not be distributed widely. 
\end{abstract}


\section{INTRODUCTION}
Epidemics, computer viruses and bushfires can all be thought of as spreading processes in which an initial localized outbreak spreads rapidly throughout a network \cite{Karafyllidis1997a,Bloem2008,Nowzari2016}. The real-world risks associated with these events have sparked significant research into methods for modeling, prediction and control. Furthermore it has stressed the importance and current limitations of methods to reduce their impact by planning appropriate intervention strategies.

The large scale these spreading processes typically evolve on, e.g. global travel networks for epidemics, large geographic areas for wildfires, or the internet for computer viruses, imposes a challenge and emphasizes the importance of scalability of computational methods. Furthermore, it can be difficult or impossible to distribute resources broadly across the network, such as vaccines for epidemics or waterbombing allocations for wildfires. Therefore, sparsity of resource allocation solutions is often required and essential.

%

Spreading processes are commonly modeled as Markov processes. The most well-known models are the Susceptible-Infected-Susceptible (SIS) model and the Susceptible-Infected-Removed (SIR) model \cite{kermark1927contributions,bailey1975mathematical}. These stochastic models can be approximated as ordinary differential equation (ODE) models, which can in turn be approximated by linear models \cite{Ahn2013,VanMieghem:2009,Nowzari2016} which is proven by \cite{VanMieghem:2009} to be an upper bound and is therefore usually the object of study.  

The problems of minimizing the spreading rate by removing either a fixed number of links or nodes in the network are both NP-hard \cite{van2011decreasing}. This fact has motivated the study of heuristics methods based on node-rankings of various forms \cite{Hadjichrysanthou2015,Dhingra2018,Lindmark}. However, in general these approaches will not be optimal in any sense and furthermore the assumption of complete link or node removal is often unrealistic.

A more realistic assumption is that spreading rate can be decreased or the recovery rate  increased by applying resources to the nodes and links. Various methods have been proposed in which this resource allocation is subject to budget constraints (e.g. \cite{Nowzari2016,Giamberardino2017,Bloem2008,Khanafer,DiGiamberardino2019,Mai2018,Liu2019b,Dangerfield2019,Torres2017,Preciado2014,Zhang2018,Han2015,Nowzari2017}). 

However, all these methods offer a static control or static intervention solution only. That is, the resource allocation is determined once based on the current state instead of allocated over time. A more realistic assumption would be to have dynamic control and hence, time dependent intervention. 

Dynamic resource allocation for epidemics is studied in \cite{drakopoulos2017network,scaman2016suppressing}. Here network properties are identified to respectively identify limitations on and derive upper and lower bounds for epidemic extinction time. More recently, dynamic programming \cite{bertsekas2000dynamic} and in particular, model predictive control (MPC) has been proposed for time dependent intervention \cite{kohler2018dynamic,watkins2019robust,zino2021analysis}. 

K{\"o}hler et al. extend their work \cite{kohler2018dynamic} in \cite{kohler2020robust} by investigating MPC approaches to optimally control the COVID-19 outbreak in Germany. Here `resource allocation' is considered social distancing measures and the goal is to minimize the number of fatalities. In \cite{Selley2015} nonlinear MPC is applied to an SIS model. The state matrix or `contact network' is altered to reduce the spread. Furthermore, critical control bounds are determined to investigate controllability. These approaches, however, do not provide sparse resource allocation. 

In this paper we, therefore, present a multi-stage sparse resource allocation method for control of spreading processes over networks that can be formulated as an efficiently solvable problem. This paper builds on work presented in \cite{LCSS2021}, where we developed sparse resource models and a node-dependent risk model for static-intervention. The presented method distinguishes itself by extending this work to a dynamic sparse multi-step approach. In particular \cite{LCSS2021} has the limitation that it takes the assumption that all resources are applied at once. In this paper we consider the more realistic scenario that resources are allocated over time, with budgets per time unit.  

The remainder of this paper is organized as follows. In Section II, we provide the problem statement. Next, in Section III, we present a relaxed problem via linearization, dynamic programming, and sparsity-inducing resource models. We demonstrate how this can be reformulated as a convex problem in Section IV, and finally, in Section V we illustrate the use of our proposed method with both epidemic and wildfire examples.

\section{PROBLEM STATEMENT}

\subsection{Spreading Process Model}
We study a spreading process on a graph $\mathcal{G}(\mathcal{V},\mathcal{E})$ with $n$ nodes in node set $\mathcal{V}$ and edge set $\mathcal{E}$, where each node $i \in \{1, 2, ..., n\} \in \mathcal{V}$ has a state $X_i(t)$ associated with it. The basic SIS model \cite{kermark1927contributions} is a cellular automata (CA) model in which a node can be in two states: infected, i.e. $X_i(t)=1$, or susceptible to infection from neighboring nodes, i.e. $X_i(t)=0$. An infected node recovers with probability $\delta_{i}\Delta t$ to $X_i(t+\Delta t)=0$ and the process spreads from infected node $j$ to susceptible node $i$ with probability $\beta_{ij} \Delta t$, where $\Delta t$ is a small time interval.

We now define $x_{i}(t)=E(X_{i}(t))=P(X_{i}(t)=1)$ as the probability of a node $i$ being infected at time $t$. Using a mean-field approximation \cite{Nowzari2016,Preciado2014} one obtains $n$ coupled nonlinear differential equations:

\begin{equation}
\label{eq:nonL}
\dot{x}_{i}=(1-x_{i}(t))\sum^{n}_{j=1}\beta_{ij}x_{j}(t)-\delta_{i}x_{i}(t).
\end{equation}
By applying Euler's forward approximation to \eqref{eq:nonL}, we obtain the following discrete time model approximation \cite{pare2020modeling}
\begin{equation}
\label{eq:disnonL}
x_{i}^{k+1}=x_i^k + h(1-x_{i}^k)\sum^{n}_{j=1}\beta_{ij}x_{j}^k-h\delta_{i}x_{i}^k
\end{equation}
where $h= t^{k+1}-t^k$ is the length of each time interval $[t^k, t^{k+1}]$. As long as $h\delta \leq 1$ and $h\sum^{n}_{j=1}\beta_{ij}<1$, this model is well-defined at $x_i^1\in[0,1] \Rightarrow x_i^k\in[0,1]$ for all $k$.

 
 
\subsection{Cost Function}
We associate the following cost function with the system  
\begin{equation}
\label{eq:JC}
  J(x^1) = \sum_{k=1}^\infty \alpha^kCx^k
\end{equation}
where $x^k=[x_{1}^k,...,x_{n}^k]^{T}$ is the state of the system \eqref{eq:disnonL} and $C= [c_1, ..., c_n]$ is a row vector defining the cost at time $k$ associated with each node $i$, with each $c_i\ge 0$. The cost is discounted by $\alpha\in(0,1]$ which can be tuned to emphasize near-term cost over long-term cost. We will also consider the \textit{risk}, or expected cost, when $x^1$ is a random variable.

\subsection{Intervention Model}
At each time step $k$ we can allocate resources to reduce the spreading rate  $\beta_{ij}$ on particular edges or increase recovery rates $\delta_{i}$ at particular nodes, in an effort to reduce the overall cost. In our model, changes made at a particular time step persist for all future times, but it is straightforward to modify our framework so that interventions have a fixed duration or diminish over time.

We assume bounded ranges of possible spreading and recovery rates:
\begin{equation}
 \label{eq:bounds2}
0 < \underline{\beta}_{ij}^k \leq \beta_{ij}^k \leq \overline{\beta}_{ij}^k, \quad 0 < \underline{\delta}_{i}^k \leq \delta_{i}^k \leq \overline{\delta}_{i}^k<\overline{\Delta}
\end{equation}
for $k=1,...,K$.
Here $\overline{\beta}_{ij}^1$ and $ \underline{\delta}_{i}^1 $ are the unmodified rates of the system \eqref{eq:disnonL}.  For $k>1$,  $\overline{\beta}_{ij}^k$ and $ \underline{\delta}_{i}^k $ are the updated rates of the system after the previous time step, i.e. $\overline{\beta}_{ij}^k=\beta_{ij}^{k-1}$ and $ \underline{\delta}_{i}^k=\delta_{i}^{k-1}$.

\subsection{Problem Statement}
\label{subsec:prob}
The goal is to keep both the risk of an outbreak of a spreading process and the allocation of resources on $\beta_{ij}$ and $\delta_i$, both per time step and overall, small. Therefore we study two closely related problems of multi-stage sparse resource allocation for spreading processes.

\begin{problem}[Resource-Constrained Risk Minimization]
Minimize the cost \eqref{eq:JC} of an outbreak of a spreading process \eqref{eq:disnonL} by changing a bounded number of elements $\beta_{ij}$ and $\delta_i$ at each time step.  
\end{problem}

\begin{problem}[Risk-Constrained Resource Minimization]
Minimize the number of non-zero resource allocations on $\beta_{ij}$ and $\delta_i$ such that a given bound on the cost \eqref{eq:JC} is achieved. 
\end{problem}

These problems are combinatorial in nature and likely to be intractable: closely related problems are NP-hard \cite{van2011decreasing}. In the remainder of this paper we introduce a tractable approximation based on a relaxed dynamic programming formulation and a sparsity-inducing resource model

\section{Dynamic Programming Relaxation}
In order to address the problems defined in the previous Section \ref{subsec:prob}, we make a number of bounding approximations to the cost of an outbreak and introduce a sparsity-inducing resource model.

\subsection{Linearized Model}
By linearizing around the infection-free equilibrium point ($x_i=0 \,\, \forall\, i$) we obtain
\begin{equation}
\label{eq:disL}
x^{k+1}=A^kx^k
\end{equation}
where $x^k=[x_{1}^k,...,x_{n}^k]^{T}$ is the state of the system and the sparse state matrix $A$ consisting of elements
 \begin{equation}
  \label{eq:epi}
 a_{ij}^k = 
\begin{cases}
1-h\delta_{i}^k  \ge 0 &\quad \text{if}\quad  i=j, \\
h\beta_{ij}^k \ge 0&\quad \text{if}\quad  i\neq j, (i,j) \in \mathcal{E}, \\
0  &\quad \text{otherwise}.
 \end {cases}
 \end{equation}
  We note that \eqref{eq:disL} is a positive system \cite{berman1994nonnegative}, i.e. if $x^1_i\ge 0$ for all $i$, then $x^i_k\ge 0$ for all $i$ and all $k\ge 0$.
 
\subsection{Cost Bounds}
We now formulate bounds on the cost \eqref{eq:JC} evaluated on solutions to the linear model \eqref{eq:disL}, which furthermore bound the costs of the nonlinear model \eqref{eq:disnonL}. It is well known (e.g. \cite{berman1994nonnegative, briat2013robust, rantzer2015scalable}) that positive systems with non-negative linear costs admit linear value (cost-to-go) functions. We will use a variation of this argument to construct bounds via a dynamic-programming-like formulation.

\begin{prop}\label{prop1}
	Suppose there exists a sequence of row vectors $p^k=[p_{1}^k,...,p_{n}^k], k = 1, ..., K$ which are elementwise non-negative: $p_i^k\ge 0$ for all $i,k$, and for which the following inequalities hold:
\begin{equation}
\label{eq:pk1}
  p^k \ge C+\alpha p^{k+1}A^{k}
\end{equation}
for $k=1,...,K-1$ and
\begin{equation}
\label{eq:pK1}
  p^K \ge C+\alpha p^K A^{K},
\end{equation}
where the inequalities are understood to hold elementwise. Then $p^1x^1\ge J(x^1)$ as defined in \eqref{eq:JC}.
\end{prop}
\begin{proof}
We first show that $V(x,k) = p^kx^k$ provides an upper bound on the cost-to-go from state $x$ at time $k$ for the linear system \eqref{eq:disL}. Since $x_i^k\ge 0$ for all $i,k$, \eqref{eq:pk1} implies
$
  p^kx^k \ge Cx^k+\alpha p^{k+1}A^{k}x^k,
$
for $k=1,...,K-1$, i.e. 
\begin{equation}
  V(x^k,k) \ge Cx^k+\alpha V(x^{k+1},k+1)
\end{equation}
and similarly \eqref{eq:pK1} implies
\begin{equation}
  V(x^k,K) \ge Cx^k+\alpha V(x^{k+1},K)
\end{equation}
for $k \ge K$. By telescoping sum we obtain
$
  V(x^1,1) \ge \sum_{k=1}^L \alpha^kCx^k + V(x^L,K)
$
for any $L\ge K$, and since $V(x^L,K)\ge 0$ and $L$ is arbitrarily large we have
$
  V(x^1,1) = p^1x^1 \ge \sum_{k=1}^\infty \alpha^kCx^k.
$

To show that his also bounds the cost for the nonlinear system, we note that for any $x^k$ with elements in $[0,1]$, $x^{k+1}\ge A^kx^k$ where $x^{k+1}$ is evaluated according to \eqref{eq:disnonL}, i.e. solutions of the linear system upper bound (elementwise) those of the nonlinear system from the same initial conditions. Since the cost vector $C$ is positive, this further implies that the cost obtained with the linear system upper bounds that with the nonlinear system.
\end{proof}

In the case that the initial condition is a random variable, we can consider the expected cost $E[J(x^1)]$. Since we have $J(x^1) \le p^1x^1$ for any $x^1$, it follows that $E[J(x^1)] \le E[p^1x^1] = p^1 E[x^1]=p^1\hat x^1$. Hence the expected cost can be bounded via a linear function of the expected value of the initial state $\hat x^1$. We, therefore, define the risk associated with each state as
\begin{equation}
\label{eq:Risk2}
   R_i=p_i^1\hat{x}_i^1.
\end{equation}
An alternative to the overall risk $\sum_i R_i=p^1\hat x^1$, is to use  $\max_{i}(R_{i})$, i.e. the worst-case expected impact of a localized outbreak.

\subsection{Resource Allocation Model}
Now that the linear spreading process dynamics and risk model are defined, we can investigate how to allocate sparse resources at each time $k$. This implies that our state matrix $A$ will change each time step and we, therefore, define $A^k$ as the state matrix with resources allocated up to time $k$.

The resource model we propose is defined by the following functions:
\begin{align}
f_{ij}\left(\beta_{ij}^k\right)&=w_{ij}\text{log}\left(\frac{\overline{\beta}^k_{ij}}{\beta_{ij}^k}\right) \nonumber \\
g_{i}\left(\overline{\Delta}-\delta_{i}^k\right)&=w_{ii}\text{log}\left(\frac{\overline{\Delta}-\underline{\delta}^k_{i}}{\overline{\Delta}-\delta_{i}^k}\right) \label{eq:RM}
\end{align}
where $w_{ij}$ are weightings expressing the cost of respectively reducing $\beta_{ij}$ and increasing $\delta_{i}$.

This type of logarithmic resource allocation was discussed in \cite{LCSS2021}. These resource models can be understood as that a certain proportional increase (for $\beta_{ij}$) and decrease (for $\delta_i$) always has the same cost. Note that it is impossible for $\beta_{ij}$ to be reduced to $0$ since that would take infinite resources. Furthermore these resource models encourage sparsity as detailed further in Section \ref{subsec:spars}.

\subsection{Optimization Formulation}
We can now formulate optimization problems to solve relaxed versions of Problem 1, Resource-Constrained Risk Minimization, and Problem 2, Risk-Constrained Resource Minimization, as defined in Section \ref{subsec:prob}.

\begin{Oproblem}
Given a defined resource allocation budget $\Gamma$, a cost $c_{i}$ associated with each node $i$, find the optimal spreading and recovery rates $\beta_{ij}^k$ and $\delta_{i}^k$ for $k=1, ..., K$ that via multi-stage sparse resource allocation minimize an upper bound on the maximum risk $p_i^1\hat{x}_i^1$. This problem can be formulated as:
\begin{align}
    \underset{p^k,\beta^k, \delta^k}{\text{minimize}} & \quad   \text{max}(p_i^1\hat{x}_i^1) \label{eq:C0}\\
    \text{such that} & \quad p^k \geq 0 \\
     & \quad \alpha p^{k+1} A^k - p^{k} \leq -C, \quad  k = 1, ..., K-1 \label{eq:DCk} \\
& \quad  p^{K} (\alpha A^K-I) \le -C \label{eq:DCK}\\
    & \quad 0 < \underline{\beta}_{ij}^k \leq \beta_{ij}^k \leq \overline{\beta}_{ij}^k, \ 0 < \underline{\delta}_{i}^k \leq \delta_{i}^k \leq \overline{\delta}_{i}^k < \overline{\Delta} \label{eq:rCons} \\
 & \quad  \sum_{ij} f_{ij}\left(\beta_{ij}^k\right) + \sum_{i} g_{i}\left(\delta_{i}^k\right) \leq \Gamma^k \label{eq:C2} \\
 & \quad \sum^K_{k=1} \left(\sum_{ij} f_{ij}\left(\beta_{ij}^k\right) + \sum_{i} g_{i}\left(\delta_{i}^k\right)\right) \leq \Gamma_{\text{tot}} \label{eq:C25}.
\end{align} 
\end{Oproblem}

\begin{Oproblem}
Find the optimal spreading and recovery rates $\beta_{ij}^k$ and $\delta_{i}^k$ for $k=1, ..., K$ that via multi-stage sparse resource allocation minimize the amount of resources required, given an upper bound on the maximum risk $\gamma$ and a cost $c_{i}$ associated with each node $i$. 

This problem can be formulated similarly except with the left-hand-side of \eqref{eq:C2} as an objective and a constraint $p_i^1\hat{x}_i^1 \leq \gamma$ for all $i$.
\end{Oproblem}

%

%

This problem formulation can be extended by including surveillance scheduling and resource allocation on $\hat{x}_i^1$, see \cite{TCNS2021} for details. Note also that we can replace \eqref{eq:C0} with $\sum_ip_i^1\hat{x}_i^1$ to bound overall risk.

\section{Exponential Cone Programming}
In this section we show that Optimization Problems 1 and 2 can be reformulated as convex optimization problems, in particular exponential cone programs. Furthermore, we discuss how the proposed resource model leads to sparse resource allocation. 

We define decision variables $y_i$ and $v_i$ for all nodes $i\in \mathcal{V}$, and $u_{ij}$ for all edges $(i,j)\in\mathcal E$. To save space we present here the convex formulation of the dynamic coupling constraints 

\begin{align}
&  \text{log} \Biggl (\sum_{i: (i,j)\in \mathcal E} \text{exp} \Biggl(y_{i}^{k+1} + \text{log}\left(\alpha h\overline{\beta}^k_{ij}\right) - \frac{u_{ij}^k}{w_{ij}} -  y_{j}^k \Biggr)  \nonumber\\ 
& + \text{exp}\left(y_{j}^{k+1} + \text{log}\left(\alpha(1-h\overline{\Delta})\right) - y_{j}^k\right)  \nonumber \\
&  + \text{exp}\left(y_{j}^{k+1}+ \text{log}\left(\alpha h(\overline{\Delta}-\underline{\delta})\right) - \frac{v_{i}^k}{w_{ii}} - y_{j}^k\right) \nonumber \\
& + \text{exp}\left (\text{log}\left (c_{j}\right) - y_{j}^k\right) \Biggr) \leq 0 \quad \forall j,  \; k = 1, ..., K-1  \label{eq:CNew} \\
&  \text{log} \Biggl (\sum_{i: (i,j)\in \mathcal E} \text{exp} \Biggl(y_{i}^{K} + \text{log}\left(\alpha h\overline{\beta}^K_{ij}\right) - \frac{u_{ij}^K}{w_{ij}} -  y_{j}^K \Biggr)  \nonumber\\ 
& + \text{exp}\left(\text{log}\left(\alpha(1-h\overline{\Delta})\right)\right)  + \text{exp}\left(\text{log}\left(\alpha h(\overline{\Delta}-\underline{\delta})\right) - \frac{v_{i}^K}{w_{ii}} \right) \nonumber \\
& + \text{exp}\left (\text{log}\left (c_{j}\right) - y_{j}^K\right) \Biggr) \leq 0 \quad \forall j,   \label{eq:C4} 
\end{align}

\begin{prop} 
Given that $h\overline{\Delta} < 1$, Optimization Problem 1 is equivalent to the following convex optimization problem under the transformations $y_i=\text{log}(p_i)$ and $u_{ij}=f_{ij}\left(\beta_{ij}\right)$ and $v_{i}=g_{i}\left(\overline{\Delta} - \delta_{i}\right)$
\begin{align}
     \underset{y^k, u^k, v^k}{\text{minimize}} & \quad  \text{max}(\text{log}(\hat{x}_i^1)+ y_i^1) \label{eq:OF1}\\
    \text{such that} & \quad \eqref{eq:CNew}, \quad \eqref{eq:C4} \nonumber \\
& \quad 0 \leq u_{ij}^k \leq w_{ij}\text{log}\left (\frac{\overline{\beta}_{ij}^k}{ \underline{\beta}_{ij}^k}\right ) \label{eq:CB} \\
& \quad 0 \leq v_{i}^k \leq w_{ii}\text{log}\left (\frac{1-\underline{\delta}^k_{i}}{ 1-\overline{\delta}^k_{i}}\right ) \label{eq:CD} \\
& \quad  \sum_{ij} u_{ij}^k + \sum_{i} v_{i}^k \leq \Gamma^k \label{eq:C5} \\
& \quad \sum^K_{k=1} \left(\sum_{ij} u_{ij}^k + \sum_{i} v_{i}^k\right) \leq \Gamma_{\text{tot}} \label{eq:C6}.
\end{align}
where $\overline{\beta}^k=\beta^{k-1}=\overline{\beta}^1-\sum^{k-1}_{k=1} u_{ij}^k$ and $\underline{\delta}^k=\delta^{k-1}=\underline{\delta}^1-\sum^{k-1}_{k=1} v_{i}^k$ for $k>1$.
\end{prop}

\begin{proof}
The objective function \eqref{eq:OF1} follows directly from \eqref{eq:C0} and $y_i=\text{log}(p_i)$. The dynamic coupling constraint for $k\neq K$ \eqref{eq:CNew} can be obtained, by rewriting \eqref{eq:DCk} as
\begin{equation}
   \sum^{n}_{i=1} \alpha p_i^{k+1} a_{ij}^k - p^k_j\leq - c_{j} \quad \forall j. 
\end{equation}
Now using \eqref{eq:epi} this can be rewritten as
\begin{equation}
\sum_{i\neq j}p_{i}^{k+1}\alpha h\beta_{ij}^k+ p^{k+1}_j\alpha\left(1-h\delta^k_j\right) -  p^k_j \leq - c_{j}  \quad \forall j
\end{equation}
which is equivalent to 
\begin{align}
\sum_{i\neq j}\frac{p_{i}^{k+1}\alpha h\beta_{ij}^k}{p^k_j}&+ \frac{p^{k+1}_j\alpha\left(1-h\Delta\right)}{p^k_j}  \nonumber \\
&+  \frac{p^{k+1}_j\alpha h(\overline{\Delta} - \delta_j^k)}{p^k_j} + \frac{c_{j}}{p^k_j} \leq 1  \quad \forall j
\end{align}
using $\overline{\Delta} > \overline{\delta}^k$. Taking the log on both sides, given that $h\overline{\Delta} < 1$, with transformation $y_i=\text{log}(p_i)$ and $u_{ij}=f_{ij}\left(\beta_{ij}\right)$ and $v_{i}=g_{i}\left(\overline{\Delta}- \delta_{i}\right)$, gives \eqref{eq:CNew}. The derivation for $k=K$ to obtain \eqref{eq:C4} from \eqref{eq:DCK} is similar, but now with $p^{k+1}=p^k$. The derivation of the bounds on the spreading rate \eqref{eq:CB} and recovery rate \eqref{eq:CD} under convex transformations $u_{ij}=f_{ij}\left(\beta_{ij}\right)$ and $v_{i}=g_{i}\left(\overline{\Delta} - \delta_{i}\right)$ can be found in \cite{LCSS2021}. Finally, equations \eqref{eq:C5} and \eqref{eq:C6} follow directly from respectively \eqref{eq:C2} and \eqref{eq:C25} and $u_{ij}=f_{ij}\left(\beta_{ij}\right)$ and $v_{i}=g_{i}\left(\overline{\Delta} - \delta_{i}\right)$. 
\end{proof}

\begin{prop} 
Given that $h\overline{\Delta} < 1$, Optimization Problem 2 is equivalent to a convex optimization problem under the transformations $y_i=\text{log}(p_i)$ and $u_{ij}=f_{ij}\left(\beta_{ij}\right)$ and $v_{i}=g_{i}\left(\overline{\Delta}- \delta_{i}\right)$. The formulation is the same as \eqref{eq:OF1} - \eqref{eq:C6} except that the objective is $\sum_{ij} u_{ij}^k + \sum_{i} v_{i}^k$ and in place of \eqref{eq:C5} we have the constraint $ \text{log}(\hat{x}_i^1)+ y_i^1 \leq \gamma$.
\end{prop}

\begin{proof}
Similar to the proof of Proposition 2.
\end{proof}

\subsection{Sparsity}
\label{subsec:spars}
A significant benefit of the the exponential cone programming formulation is that it encourages sparse resource allocation. Our resource models, now formulated as constraint \eqref{eq:C5} are $\ell_{1}$ norm constraints and objectives, since $u_{ij}\geq 0$ and $v_i \geq 0$, which are widely used to encourage sparsity \cite{tibshirani1996regression,candes2006robust,donoho2006compressed,Candes2008}.

If the goal is maximal sparsity, i.e. minimal number of nodes and edges with non-zero resources allocated, then we can use the reweighted $\ell_{1}$ optimization approach of \cite{Candes2008}.  We can adapt this to our problem by iteratively solving the problem, but with a reweighted resource model that approximates the number of nodes and edges with non-zero allocation, e.g.
\begin{equation}
\label{eq:L1} 
\phi^q=\sum_{ij} \frac{u^{k^q}_{ij}}{u_{ij}^{k^{q-1}}+\epsilon}+\sum_{i} \frac{v^{k^q}_{i}}{v_{i}^{k^{q-1}}+\epsilon}
\end{equation} 
where $q$ is the iteration number and $\epsilon$ a small positive constant to improve numerical stability. If we use our resource model as a constraint in the optimization problem, we now replace the constraint with $\phi^q \leq M$ where $M$ is the bound on the number of nodes and edges that can have resources allocated to them. This iteration has no guarantee of convergence or global optimality, but has been found to be very effective in practice.

\section{RESULTS}
In this section we illustrate how our proposed method can be implemented with examples for two different spreading processes; epidemics and wildfires. First, we discuss a 7 node epidemic example to demonstrate how the optimization framework can be used for deriving a vaccination strategy. Second, we demonstrate with a 1000 node wildfire example that our multi-stage approach provides sparse results.

\subsection{Epidemic Example}
To demonstrate how the optimization framework can be utilized we take an example of an epidemic spreading on a graph $\mathcal{G}(\mathcal{V},\mathcal{E})$ with $n=7$ nodes as visualized in Fig. \ref{fig:GraphAM1}. Here, node color indicates estimated outbreak probability $\hat{x}_i^1$, weightings $w_{ij}$ are indicated at the edges and node marker size indicates cost $c_i$. We take homogeneous spreading and recovery rates of $\beta_{ij}=0.35$ for all $(i,j) \in \mathcal{E}$ and $\delta=0.2$ for all $i \in \mathcal{V}$. In this particular example we have a vulnerable high cost node $i=7$, which could be an elderly person, which is connected to $i=6$. Node $i=6$ is connected to node $i=1$, who has a high likelihood of getting the disease, e.g. due to their work situation.

\begin{figure}[!t]
\centering
 \def\svgwidth{0.75\linewidth}
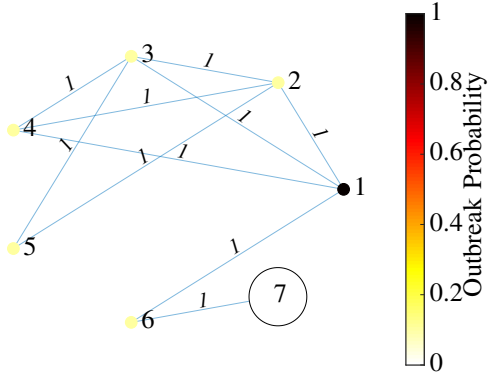    
\caption{Epidemic spreading over a graph with $n=7$ nodes. Node color indicates outbreak probability $\hat{x}_i^1$, weightings $w_{ij}$ are indicated at the edges, node marker size indicates cost $c_i$.}
\label{fig:GraphAM1}
\end{figure}

We, now, want to use our multi-stage approach to obtain a vaccination strategy. We model vaccination similar to other common methods, where vaccination of a node increases both its recovery rate \cite{forster2007optimizing,hansen2011optimal,Han2015} and decreases all incoming spreading rates \cite{zaman2008stability,hethcote2000mathematics,chen2006susceptible,kar2011stability}. Setting $\overline{\Delta}=1$ and taking $\alpha=0.93$, $h=0.24$ and $\Gamma^k=1.5$ for all $k$ with $K=4$, the obtained resource allocation, or vaccination strategy, is given in Fig. \ref{fig:03}. 

\begin{figure}
\centering
\begin{subfigure}[b]{0.45\linewidth}
         \def\svgwidth{1\textwidth}
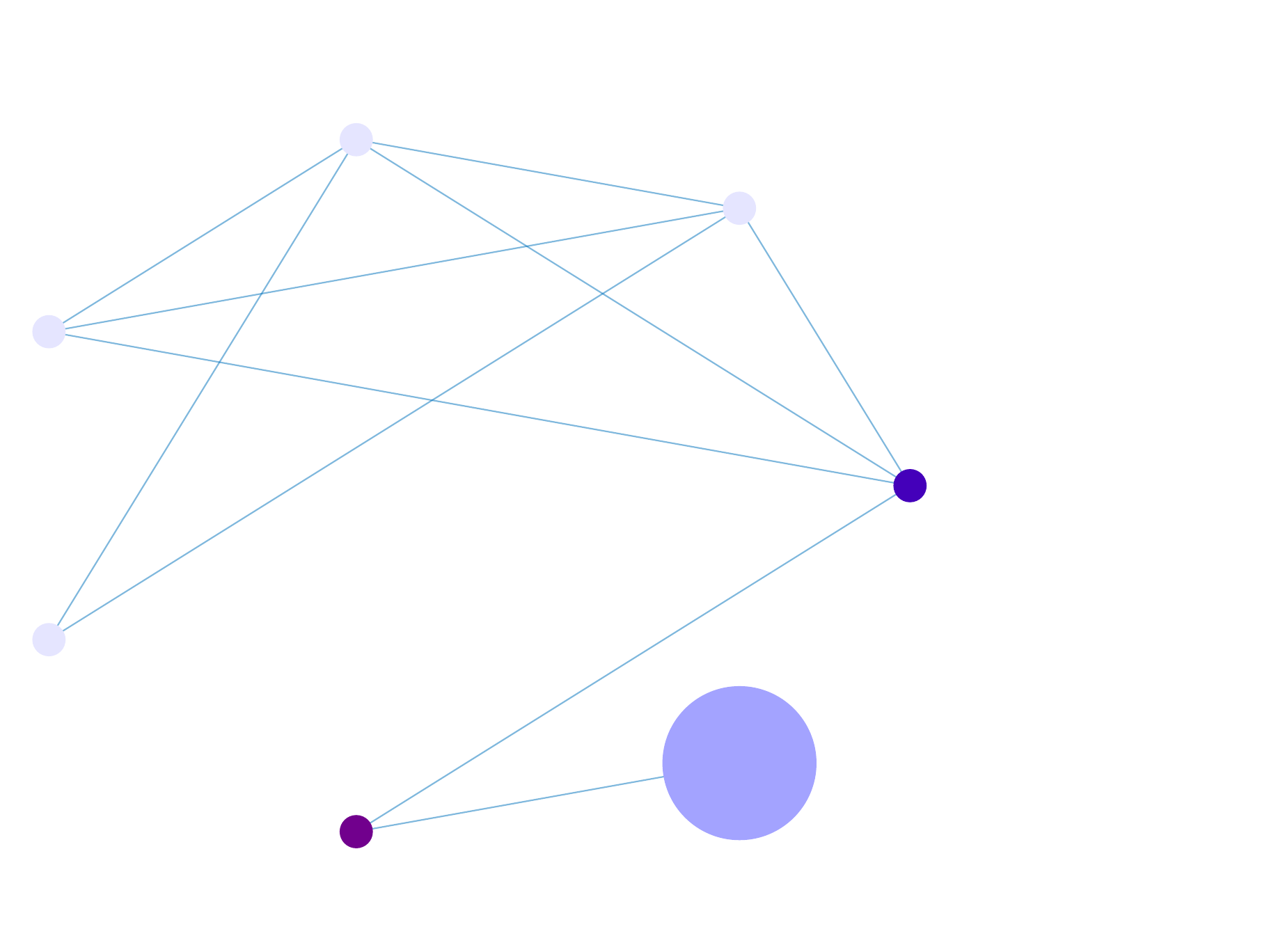  
      \caption{k=1}
      \label{fig:031}
  \end{subfigure}
  ~ 
     \begin{subfigure}[b]{0.45\linewidth}
         \def\svgwidth{1\textwidth}
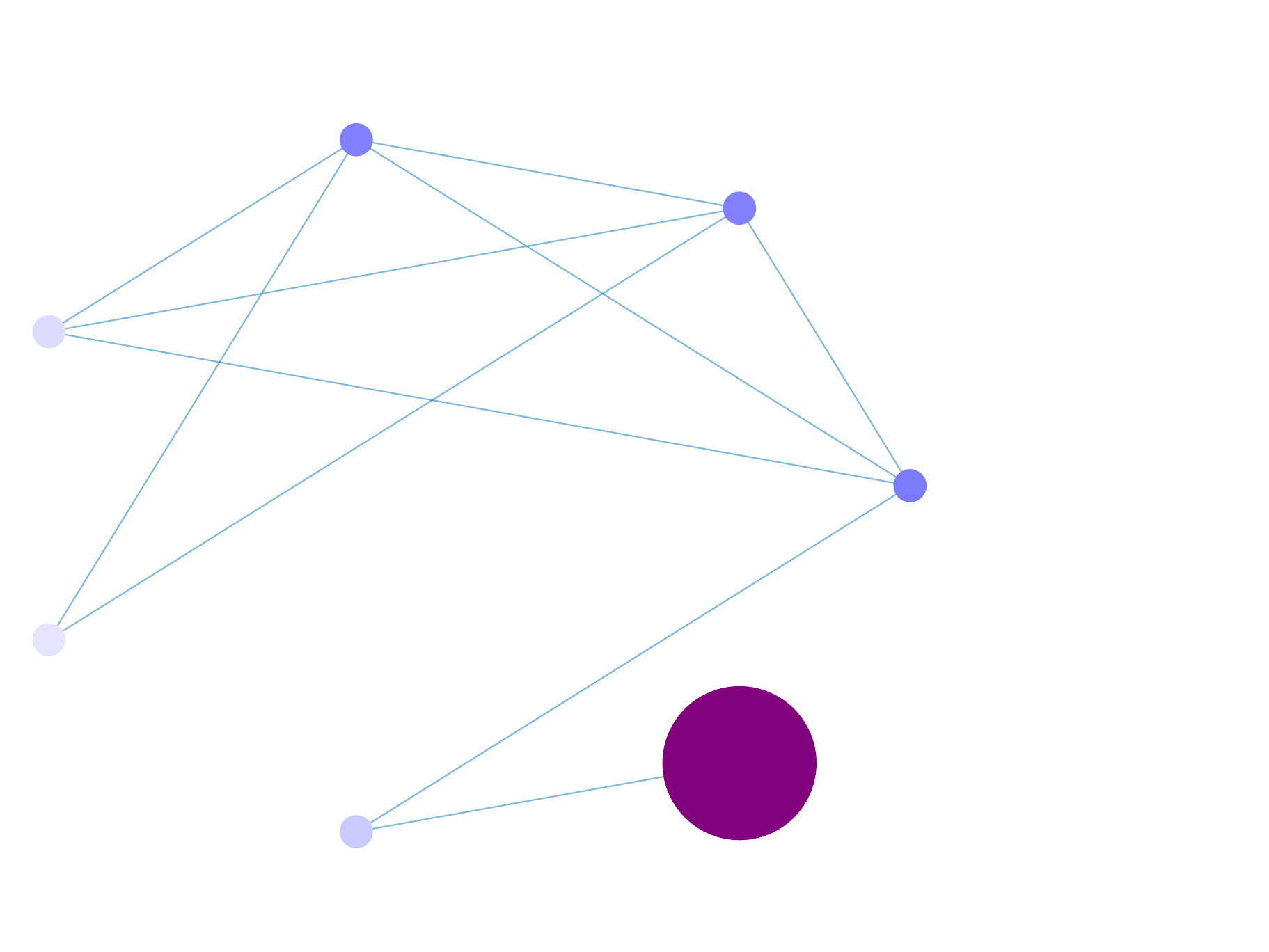 
\caption{k=2}
      \label{fig:032}
  \end{subfigure}
  \begin{subfigure}[b]{0.45\linewidth}
         \def\svgwidth{1\textwidth}
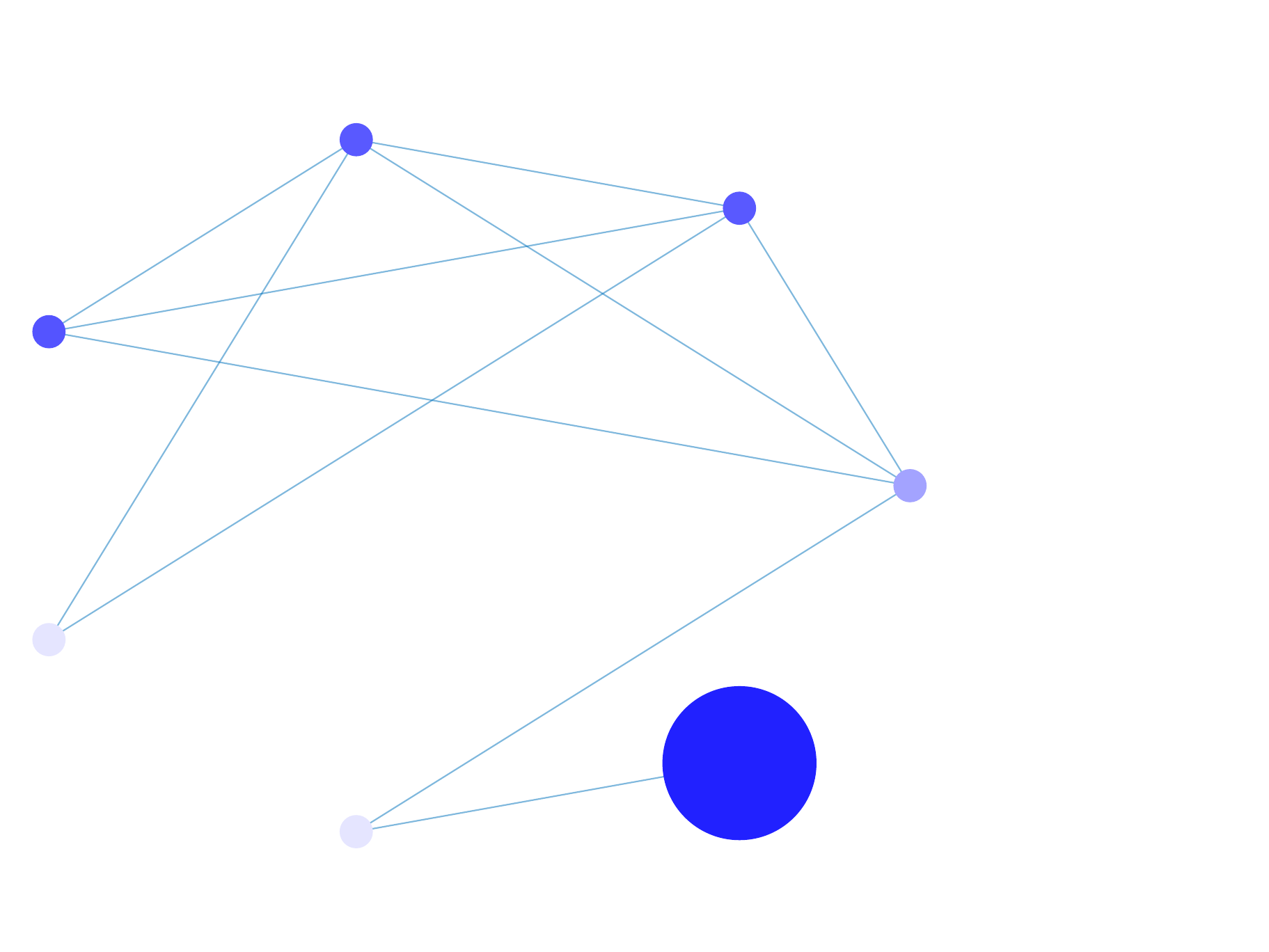 
      \caption{k=3}
      \label{fig:033}
  \end{subfigure}
  ~ 
     \begin{subfigure}[b]{0.45\linewidth}
         \def\svgwidth{1\textwidth}
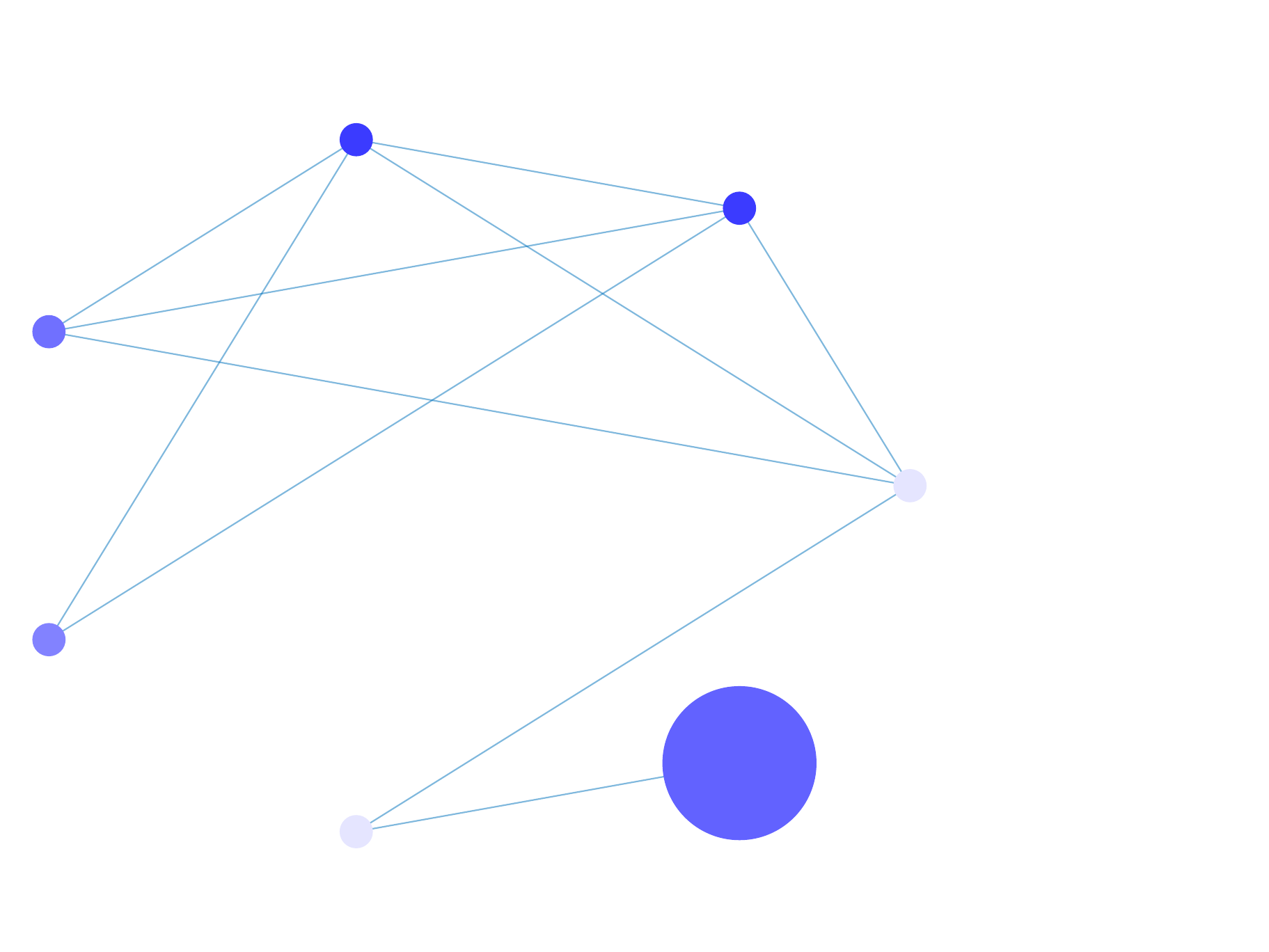 
\caption{k=K}
      \label{fig:034}
  \end{subfigure}
      \caption{Vaccination strategy for a 4-stage approach with $\Gamma^k=1.5$, $h=0.24$, $\alpha=0.93$. Node color indicates nodes vaccinated.}     
\label{fig:03}
\end{figure}

It can be seen that we should first vaccinate node $i=1$, with a high likelihood of the disease, and node $i=6$, connected to the vulnerable node $i=7$, which is vaccinated next, before moving on to the rest of the population. Decreasing the interval $h= t^{k+1}-t^k$ would result in a similar vaccination strategy, but with a slower roll out, i.e. more time steps are needed to get to the general population.

To further improve sparsity we take the reweighted $\ell_1$ minimization as explained in Section \ref{subsec:spars}. We now solve, however, Optimization Problem 2, keeping the same risk bound as obtained for Fig. \ref{fig:03}, but minimize the amount of non-zero allocations. The results are shown in Fig. \ref{fig:Ex1RL1}. The solution is to vaccinate the node with a high likelihood $i=1$ and next the vulnerable node $i=7$ before vaccinating 2 members of the general population.  

\begin{figure}
\centering
\begin{subfigure}[b]{0.45\linewidth}
         \def\svgwidth{1\textwidth}
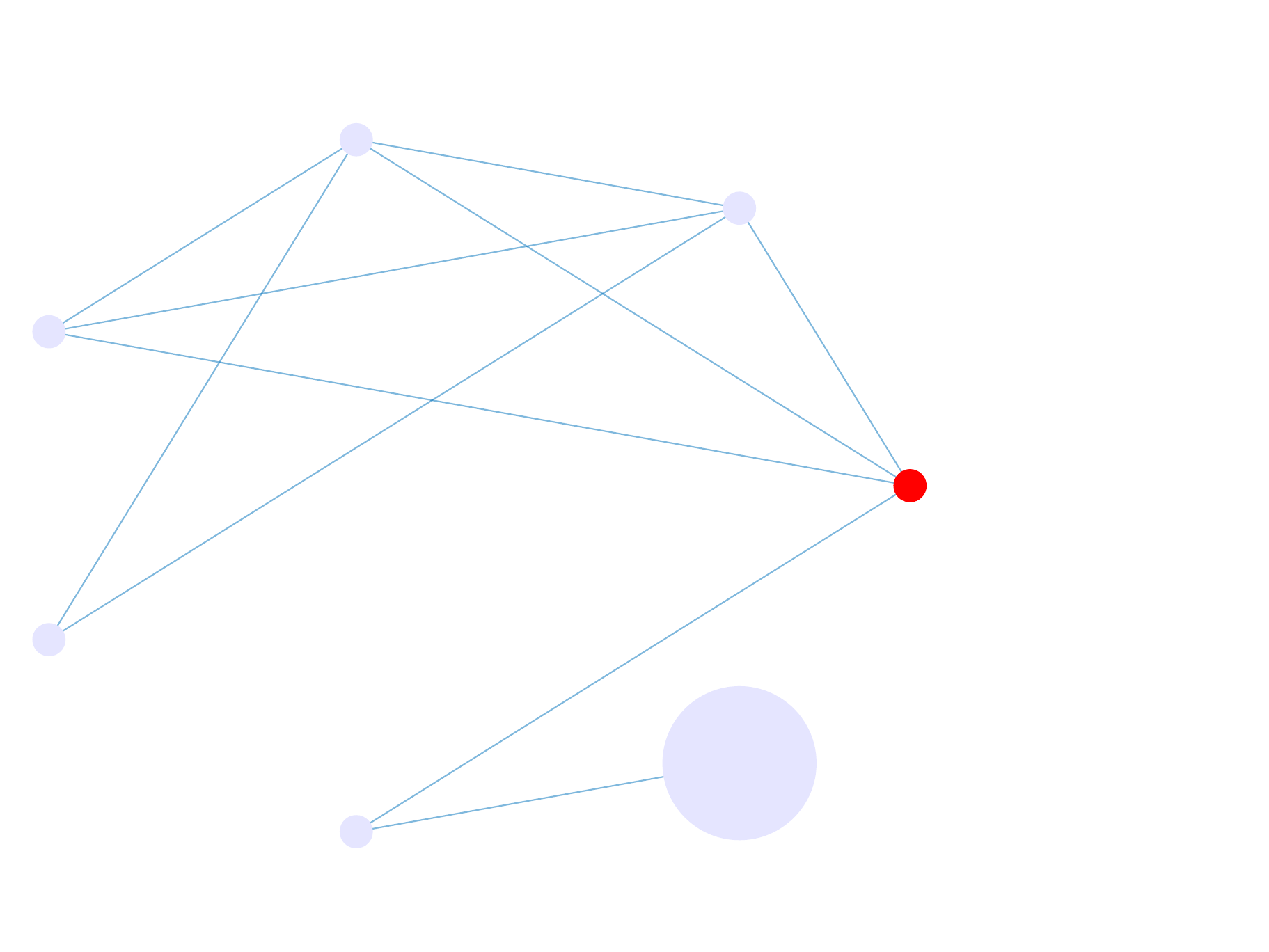  
      \caption{k=1}
      \label{fig:Ex1Ra}
  \end{subfigure}
  ~ 
     \begin{subfigure}[b]{0.45\linewidth}
         \def\svgwidth{1\textwidth}
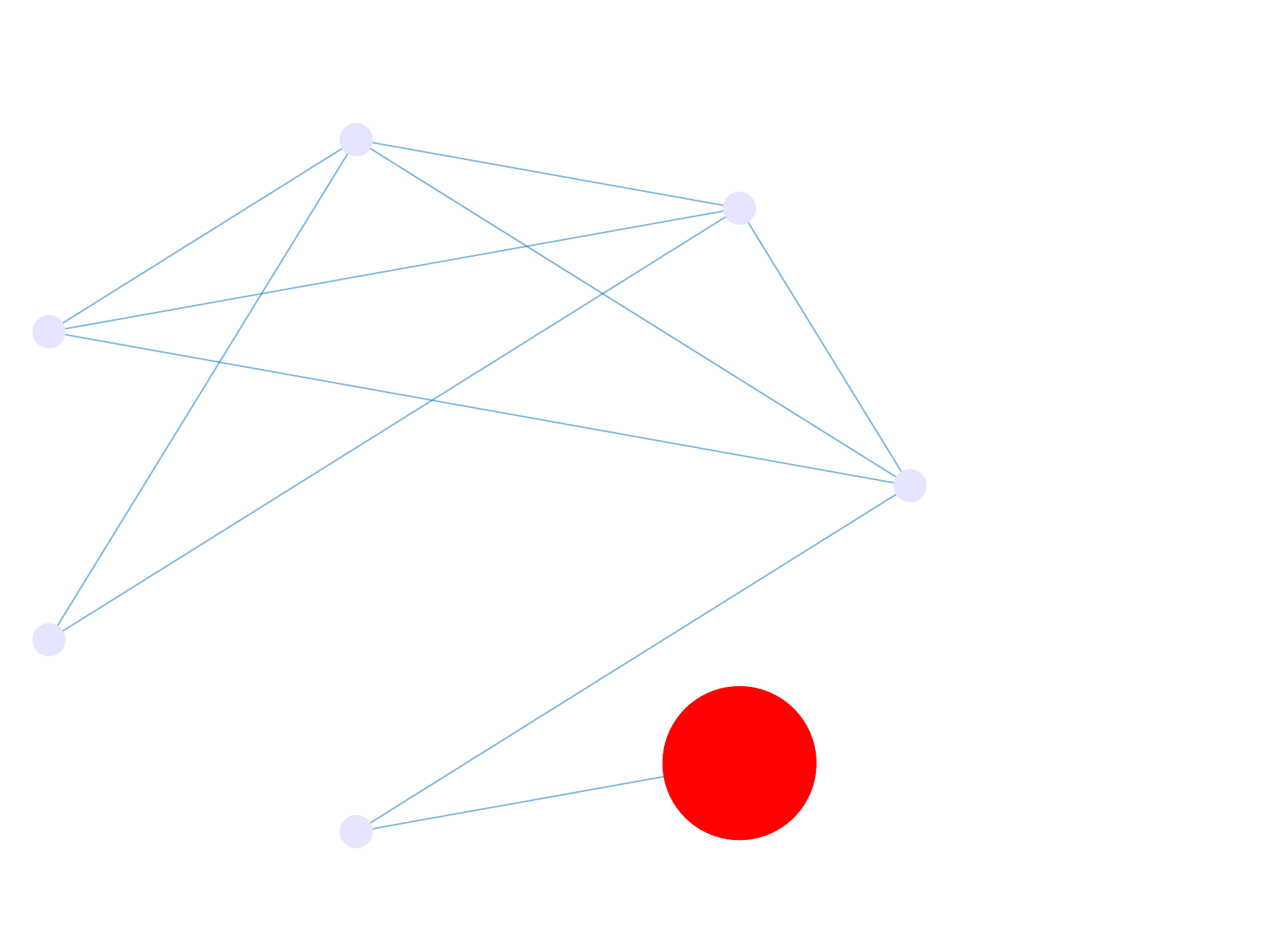 
\caption{k=2}
      \label{fig:Ex1Rb}
  \end{subfigure}
  \begin{subfigure}[b]{0.45\linewidth}
         \def\svgwidth{1\textwidth}
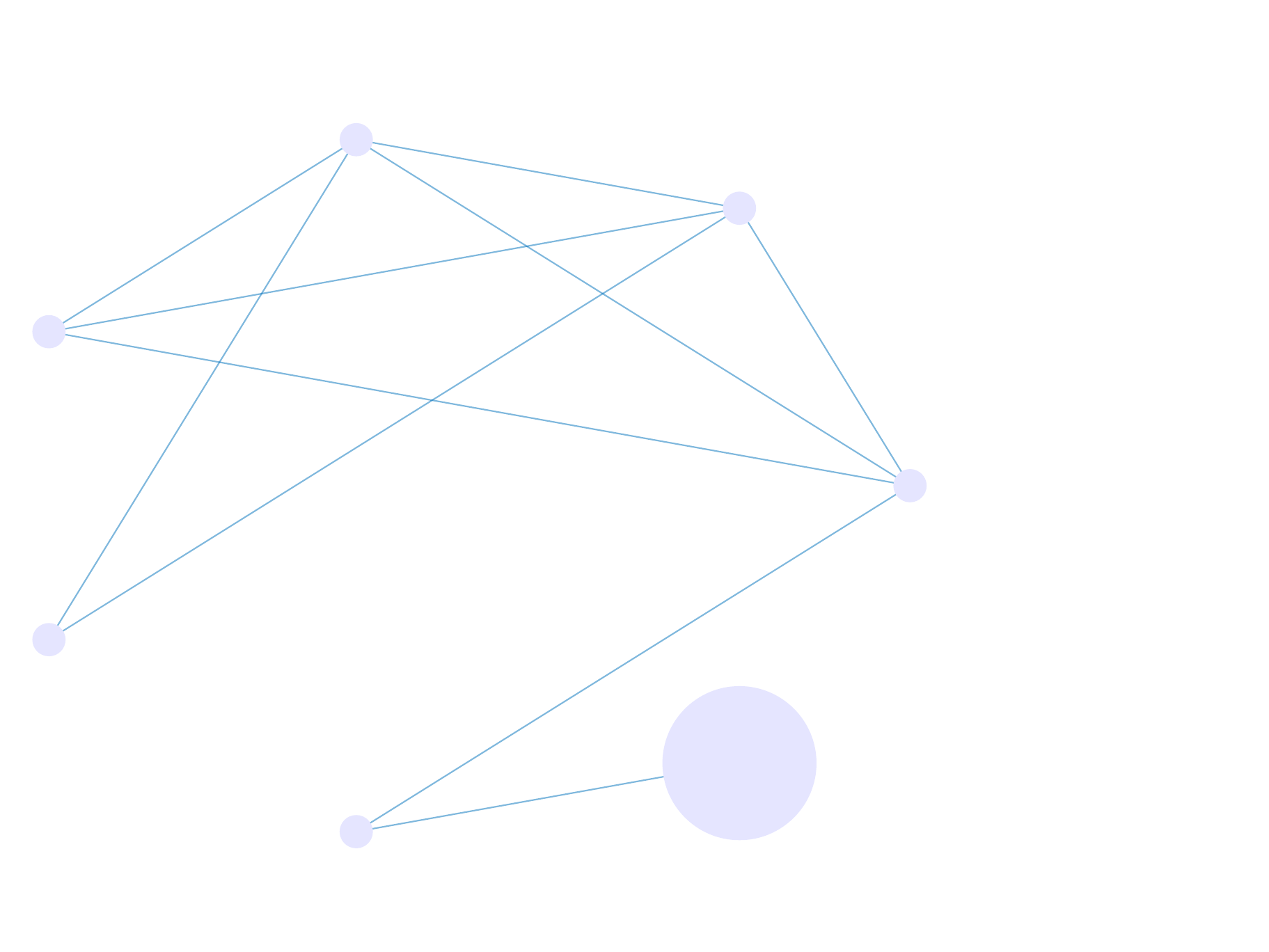  
      \caption{k=3}
      \label{fig:Ex1Rc}
  \end{subfigure}
  ~ 
     \begin{subfigure}[b]{0.45\linewidth}
         \def\svgwidth{1\textwidth}
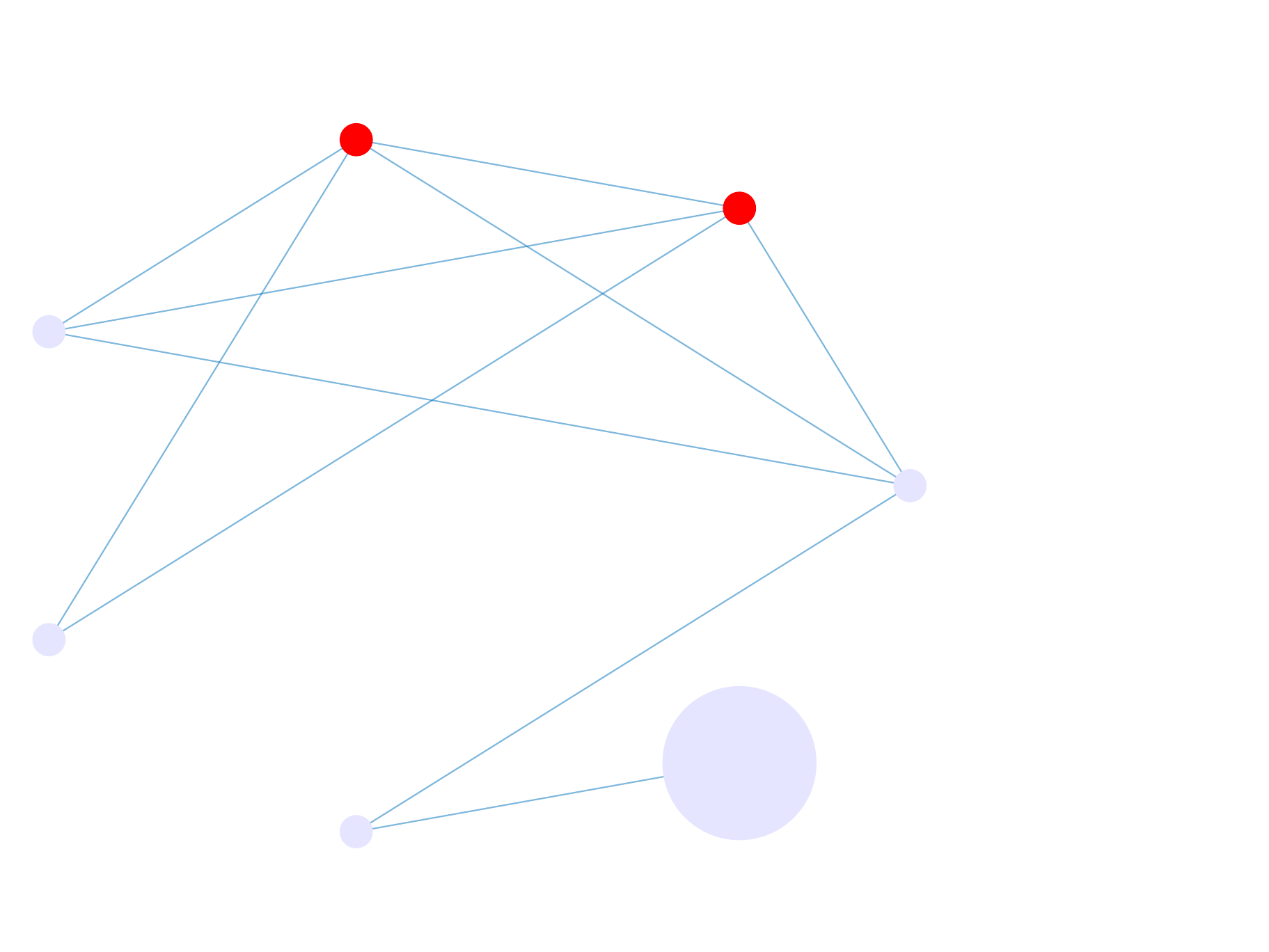 
\caption{k=K}
      \label{fig:Ex1Rd}
  \end{subfigure}
      \caption{Reweighted vaccination strategy of Fig. \ref{fig:03}}.      
\label{fig:Ex1RL1}
\end{figure}

\subsection{Wildfire Example}
Let us consider the fictional landscape given in Fig. \ref{fig:Landscape} consisting of different vegetation types, a city and water. We represent this landscape as a graph $\mathcal{G}(\mathcal{V},\mathcal{E})$ with $n=1000$ nodes, where the edge set $\mathcal{E}$ connects each node to its neighboring 8 nodes, i.e. fire can spread horizontally, vertically and diagonally. Boundary nodes have fewer edges accordingly.

\begin{figure}
\centering
\begin{subfigure}[b]{0.9\linewidth}
\def\svgwidth{1\textwidth}
\includegraphics[width=1\linewidth]{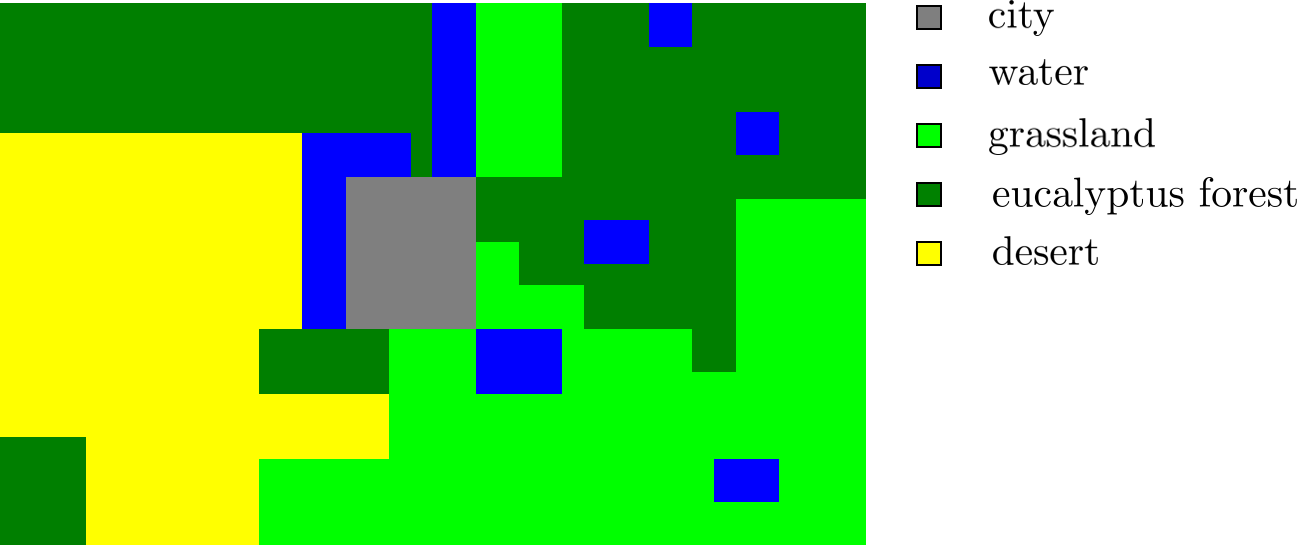}
      \caption{Vegetation Map}
\label{fig:Landscape}
\end{subfigure}
\begin{subfigure}[b]{0.9\linewidth}
\def\svgwidth{0.79\textwidth}
\begingroup%
  \makeatletter%
  \providecommand\color[2][]{%
    \errmessage{(Inkscape) Color is used for the text in Inkscape, but the package 'color.sty' is not loaded}%
    \renewcommand\color[2][]{}%
  }%
  \providecommand\transparent[1]{%
    \errmessage{(Inkscape) Transparency is used (non-zero) for the text in Inkscape, but the package 'transparent.sty' is not loaded}%
    \renewcommand\transparent[1]{}%
  }%
  \providecommand\rotatebox[2]{#2}%
  \newcommand*\fsize{\dimexpr\f@size pt\relax}%
  \newcommand*\lineheight[1]{\fontsize{\fsize}{#1\fsize}\selectfont}%
  \ifx\svgwidth\undefined%
    \setlength{\unitlength}{513.52734375bp}%
    \ifx\svgscale\undefined%
      \relax%
    \else%
      \setlength{\unitlength}{\unitlength * \real{\svgscale}}%
    \fi%
  \else%
    \setlength{\unitlength}{\svgwidth}%
  \fi%
  \global\let\svgwidth\undefined%
  \global\let\svgscale\undefined%
  \makeatother%
  \begin{picture}(1,0.54973215)%
    \lineheight{1}%
    \setlength\tabcolsep{0pt}%
    \put(0,0){\includegraphics[width=\unitlength,page=1]{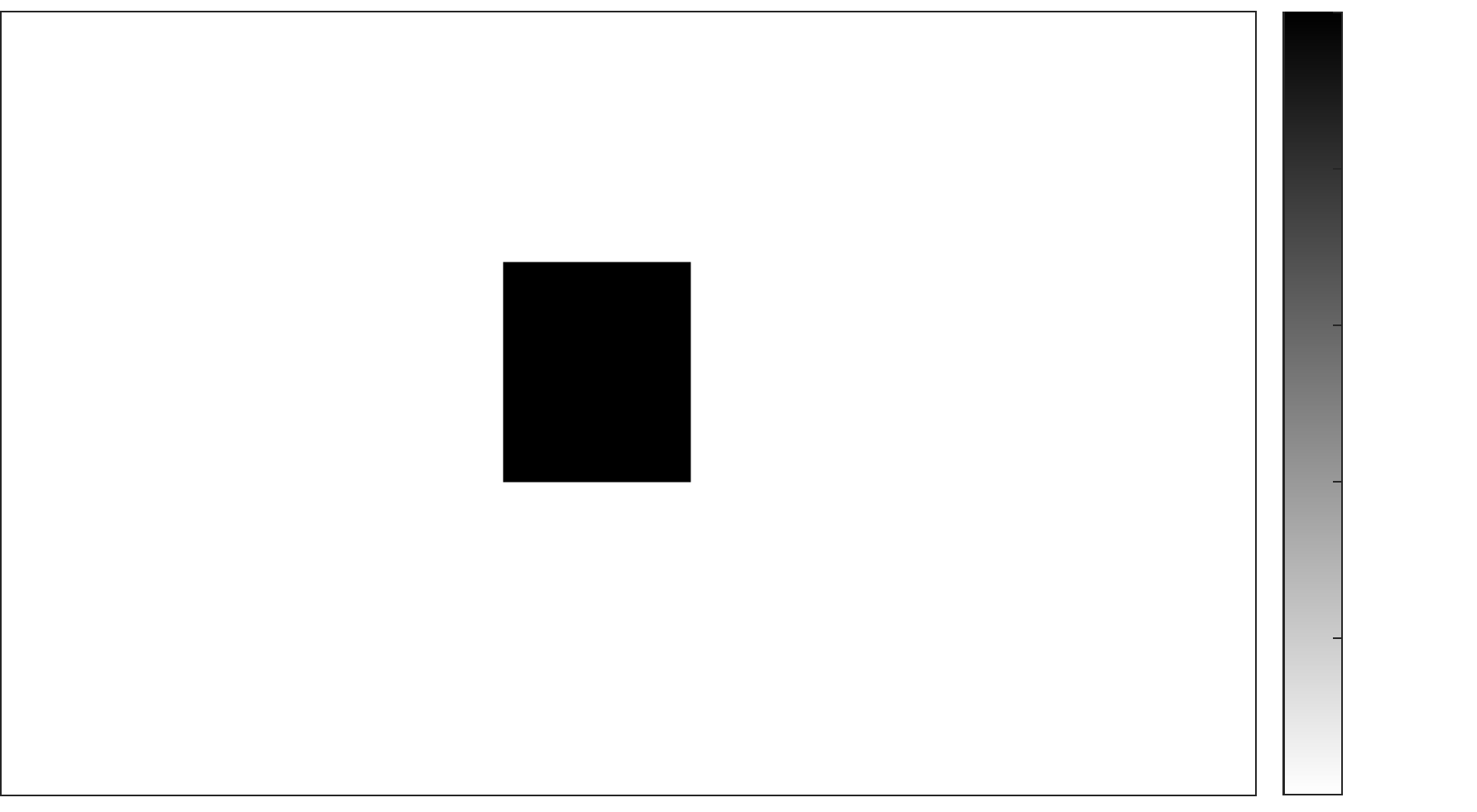}}%
    \put(0.91915741,0.00039793){\makebox(0,0)[lt]{\lineheight{1.25}\smash{\begin{tabular}[t]{l}\scriptsize 0\end{tabular}}}}%
    \put(0.91915741,0.10642928){\makebox(0,0)[lt]{\lineheight{1.25}\smash{\begin{tabular}[t]{l}\scriptsize 0.2\end{tabular}}}}%
    \put(0.91915741,0.21246064){\makebox(0,0)[lt]{\lineheight{1.25}\smash{\begin{tabular}[t]{l}\scriptsize 0.4\end{tabular}}}}%
    \put(0.91915741,0.31849199){\makebox(0,0)[lt]{\lineheight{1.25}\smash{\begin{tabular}[t]{l}\scriptsize 0.6\end{tabular}}}}%
    \put(0.91915741,0.42452335){\makebox(0,0)[lt]{\lineheight{1.25}\smash{\begin{tabular}[t]{l}\scriptsize 0.8\end{tabular}}}}%
    \put(0.91915741,0.5305547){\makebox(0,0)[lt]{\lineheight{1.25}\smash{\begin{tabular}[t]{l}\scriptsize 1\end{tabular}}}}%
    \put(0.99999,0.24137857){\rotatebox{90}{\makebox(0,0)[lt]{\lineheight{1.25}\smash{\begin{tabular}[t]{l}\small Cost\end{tabular}}}}}%
    \put(0,0){\includegraphics[width=\unitlength,page=2]{Ex2Cost.pdf}}%
  \end{picture}%
\endgroup%
   
      \caption{Cost Map}
      \label{fig:53Cost}
  \end{subfigure}
  ~ 
     \begin{subfigure}[b]{0.9\linewidth}
\def\svgwidth{0.79\textwidth}
\begingroup%
  \makeatletter%
  \providecommand\color[2][]{%
    \errmessage{(Inkscape) Color is used for the text in Inkscape, but the package 'color.sty' is not loaded}%
    \renewcommand\color[2][]{}%
  }%
  \providecommand\transparent[1]{%
    \errmessage{(Inkscape) Transparency is used (non-zero) for the text in Inkscape, but the package 'transparent.sty' is not loaded}%
    \renewcommand\transparent[1]{}%
  }%
  \providecommand\rotatebox[2]{#2}%
  \newcommand*\fsize{\dimexpr\f@size pt\relax}%
  \newcommand*\lineheight[1]{\fontsize{\fsize}{#1\fsize}\selectfont}%
  \ifx\svgwidth\undefined%
    \setlength{\unitlength}{516.21487427bp}%
    \ifx\svgscale\undefined%
      \relax%
    \else%
      \setlength{\unitlength}{\unitlength * \real{\svgscale}}%
    \fi%
  \else%
    \setlength{\unitlength}{\svgwidth}%
  \fi%
  \global\let\svgwidth\undefined%
  \global\let\svgscale\undefined%
  \makeatother%
  \begin{picture}(1,0.54687012)%
    \lineheight{1}%
    \setlength\tabcolsep{0pt}%
    \put(0,0){\includegraphics[width=\unitlength,page=1]{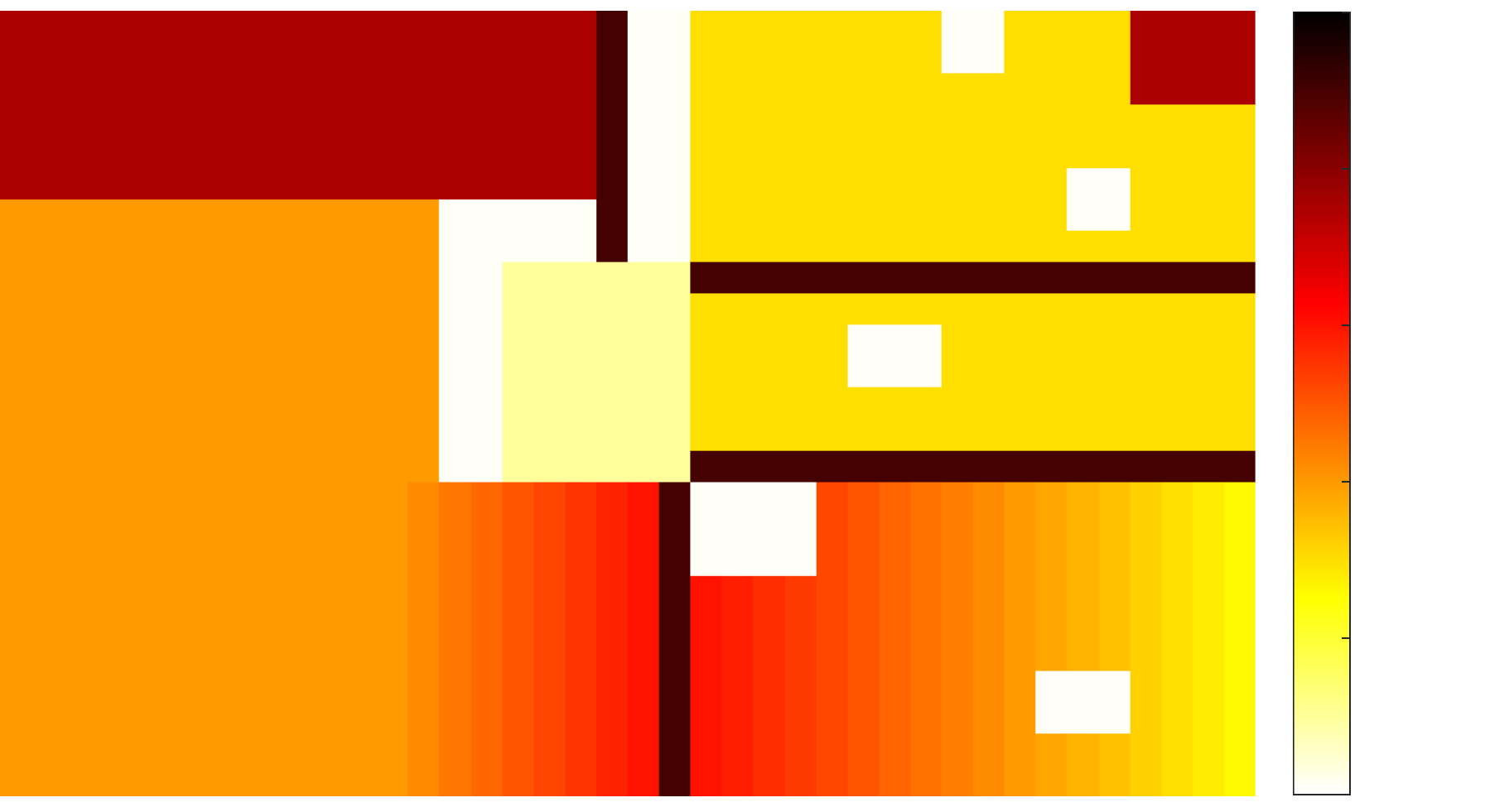}}%
    \put(0.91957826,0.00039585){\makebox(0,0)[lt]{\lineheight{1.25}\smash{\begin{tabular}[t]{l}\scriptsize0\end{tabular}}}}%
    \put(0.91957826,0.10587519){\makebox(0,0)[lt]{\lineheight{1.25}\smash{\begin{tabular}[t]{l}\scriptsize0.2\end{tabular}}}}%
    \put(0.91957826,0.21135452){\makebox(0,0)[lt]{\lineheight{1.25}\smash{\begin{tabular}[t]{l}\scriptsize0.4\end{tabular}}}}%
    \put(0.91957826,0.31683385){\makebox(0,0)[lt]{\lineheight{1.25}\smash{\begin{tabular}[t]{l}\scriptsize0.6\end{tabular}}}}%
    \put(0.91957826,0.42231318){\makebox(0,0)[lt]{\lineheight{1.25}\smash{\begin{tabular}[t]{l}\scriptsize0.8\end{tabular}}}}%
    \put(0.91957826,0.52779251){\makebox(0,0)[lt]{\lineheight{1.25}\smash{\begin{tabular}[t]{l}\scriptsize1\end{tabular}}}}%
    \put(0.99999,0.05457196){\rotatebox{90}{\makebox(0,0)[lt]{\lineheight{1.25}\smash{\begin{tabular}[t]{l}\small Outbreak Probability\end{tabular}}}}}%
    \put(0,0){\includegraphics[width=\unitlength,page=2]{Ex2L.pdf}}%
  \end{picture}%
\endgroup%
   
      \caption{Outbreak Probability}
      \label{fig:LM}
  \end{subfigure}
      \caption{Fictional landscape with different area types, represented as a grid with $n=1000$ nodes with its corresponding cost map and estimated outbreak probability. City nodes have a high cost associated with them.}
\label{fig:53}
\end{figure}

The spreading rates are determined by the vegetation type, wind speed and direction, as per the wildfire models in \cite{Karafyllidis1997a} and \cite{Alexandridis2008a}. The spreading rate for an edge $\beta=\beta_{b}\beta_{veg}\beta_{w}$ consists of a baseline spreading rate $\beta_{b}=0.5$ and $ \beta_{veg}=0.1, 1$ and $1.4$ for respectively desert, grassland and eucalypt forest. Water is considered unburnable and those edges are removed. $\beta_{w}$ is calculated following \cite{Alexandridis2008a} for a northeasterly wind of $V=4$ m/s. The spreading rate is corrected for the spread between diagonally connected nodes, following \cite{Karafyllidis1997a}. We set the recovery rate $\delta=0.5$ for all nodes $i \in \mathcal{V}$.

We consider resource allocation to reduce the spreading rates $\beta_{ij}$. We set $w_{ij}=1$ for all edges $(i,j) \in \mathcal{E}$, indicating an equal cost to apply resources to any edge. The cost of the city nodes is taken as $c_{i}=1$, whereas $c_{i}=0.001$ for all other nodes as illustrated in Fig. \ref{fig:53Cost}. This could reflect either the higher economic cost of fire reaching a city or the higher risk to human life. The outbreak probability $\hat{x}_{i}^1$ is depicted in Fig. \ref{fig:LM}, e.g. this could reflect that fires are more likely to start near roads. The discount factor is set to $\alpha=0.9$. 

We now take a multi-stage approach with 4 steps, i.e. $K=4$, and set $h=0.036$. The results are illustrated in Fig. \ref{fig:FG10A}. The accumulated amount of resources are shown, where the new resources added per time step have a higher line width and previous resources added a smaller line width and higher transparency. It can be seen that effectively a wall is built around the city: in the first stage, resources are allocated to stop the spread from in particular eucalypt areas in the north, with a high spreading rate and high outbreak probability, to the city. In the next stages, resources are invested on the east and south to guard against spread from those areas.  

\begin{figure*}
\centering
\begin{subfigure}[b]{0.45\linewidth}
        \def\svgwidth{1\textwidth}
\begingroup%
  \makeatletter%
  \providecommand\color[2][]{%
    \errmessage{(Inkscape) Color is used for the text in Inkscape, but the package 'color.sty' is not loaded}%
    \renewcommand\color[2][]{}%
  }%
  \providecommand\transparent[1]{%
    \errmessage{(Inkscape) Transparency is used (non-zero) for the text in Inkscape, but the package 'transparent.sty' is not loaded}%
    \renewcommand\transparent[1]{}%
  }%
  \providecommand\rotatebox[2]{#2}%
  \newcommand*\fsize{\dimexpr\f@size pt\relax}%
  \newcommand*\lineheight[1]{\fontsize{\fsize}{#1\fsize}\selectfont}%
  \ifx\svgwidth\undefined%
    \setlength{\unitlength}{502.61453247bp}%
    \ifx\svgscale\undefined%
      \relax%
    \else%
      \setlength{\unitlength}{\unitlength * \real{\svgscale}}%
    \fi%
  \else%
    \setlength{\unitlength}{\svgwidth}%
  \fi%
  \global\let\svgwidth\undefined%
  \global\let\svgscale\undefined%
  \makeatother%
  \begin{picture}(1,0.57360556)%
    \lineheight{1}%
    \setlength\tabcolsep{0pt}%
    \put(0,0){\includegraphics[width=\unitlength,page=1]{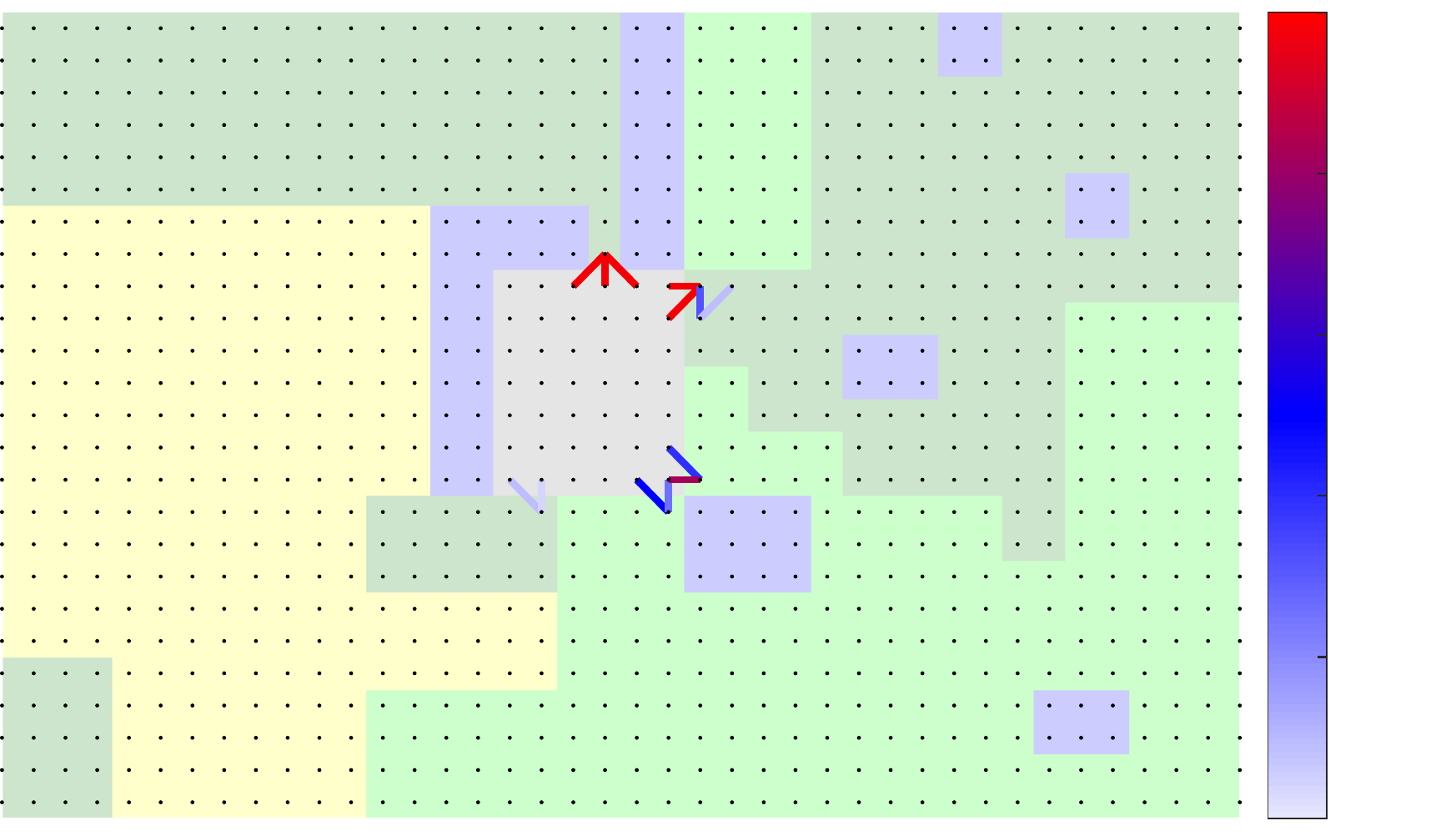}}%
    \put(0.92195404,0.00040657){\makebox(0,0)[lt]{\lineheight{1.25}\smash{\begin{tabular}[t]{l}\footnotesize 0\end{tabular}}}}%
    \put(0.92195404,0.1111276){\makebox(0,0)[lt]{\lineheight{1.25}\smash{\begin{tabular}[t]{l}\footnotesize 0.2\end{tabular}}}}%
    \put(0.92195404,0.22184863){\makebox(0,0)[lt]{\lineheight{1.25}\smash{\begin{tabular}[t]{l}\footnotesize 0.4\end{tabular}}}}%
    \put(0.92195404,0.33256966){\makebox(0,0)[lt]{\lineheight{1.25}\smash{\begin{tabular}[t]{l}\footnotesize 0.6\end{tabular}}}}%
    \put(0.92195404,0.4432907){\makebox(0,0)[lt]{\lineheight{1.25}\smash{\begin{tabular}[t]{l}\footnotesize0.8\end{tabular}}}}%
    \put(0.92195404,0.55401173){\makebox(0,0)[lt]{\lineheight{1.25}\smash{\begin{tabular}[t]{l}\footnotesize1\end{tabular}}}}%
    \put(0.99999,0.11408757){\rotatebox{90}{\makebox(0,0)[lt]{\lineheight{1.25}\smash{\begin{tabular}[t]{l}Resource allocation\end{tabular}}}}}%
    \put(0,0){\includegraphics[width=\unitlength,page=2]{Ex2aN.pdf}}%
  \end{picture}%
\endgroup%
   
      \caption{k=1}
      \label{fig:FG10aA}
  \end{subfigure}
  ~
\begin{subfigure}[b]{0.45\linewidth}
        \def\svgwidth{1\textwidth}
\begingroup%
  \makeatletter%
  \providecommand\color[2][]{%
    \errmessage{(Inkscape) Color is used for the text in Inkscape, but the package 'color.sty' is not loaded}%
    \renewcommand\color[2][]{}%
  }%
  \providecommand\transparent[1]{%
    \errmessage{(Inkscape) Transparency is used (non-zero) for the text in Inkscape, but the package 'transparent.sty' is not loaded}%
    \renewcommand\transparent[1]{}%
  }%
  \providecommand\rotatebox[2]{#2}%
  \newcommand*\fsize{\dimexpr\f@size pt\relax}%
  \newcommand*\lineheight[1]{\fontsize{\fsize}{#1\fsize}\selectfont}%
  \ifx\svgwidth\undefined%
    \setlength{\unitlength}{502.61453247bp}%
    \ifx\svgscale\undefined%
      \relax%
    \else%
      \setlength{\unitlength}{\unitlength * \real{\svgscale}}%
    \fi%
  \else%
    \setlength{\unitlength}{\svgwidth}%
  \fi%
  \global\let\svgwidth\undefined%
  \global\let\svgscale\undefined%
  \makeatother%
  \begin{picture}(1,0.57360556)%
    \lineheight{1}%
    \setlength\tabcolsep{0pt}%
    \put(0,0){\includegraphics[width=\unitlength,page=1]{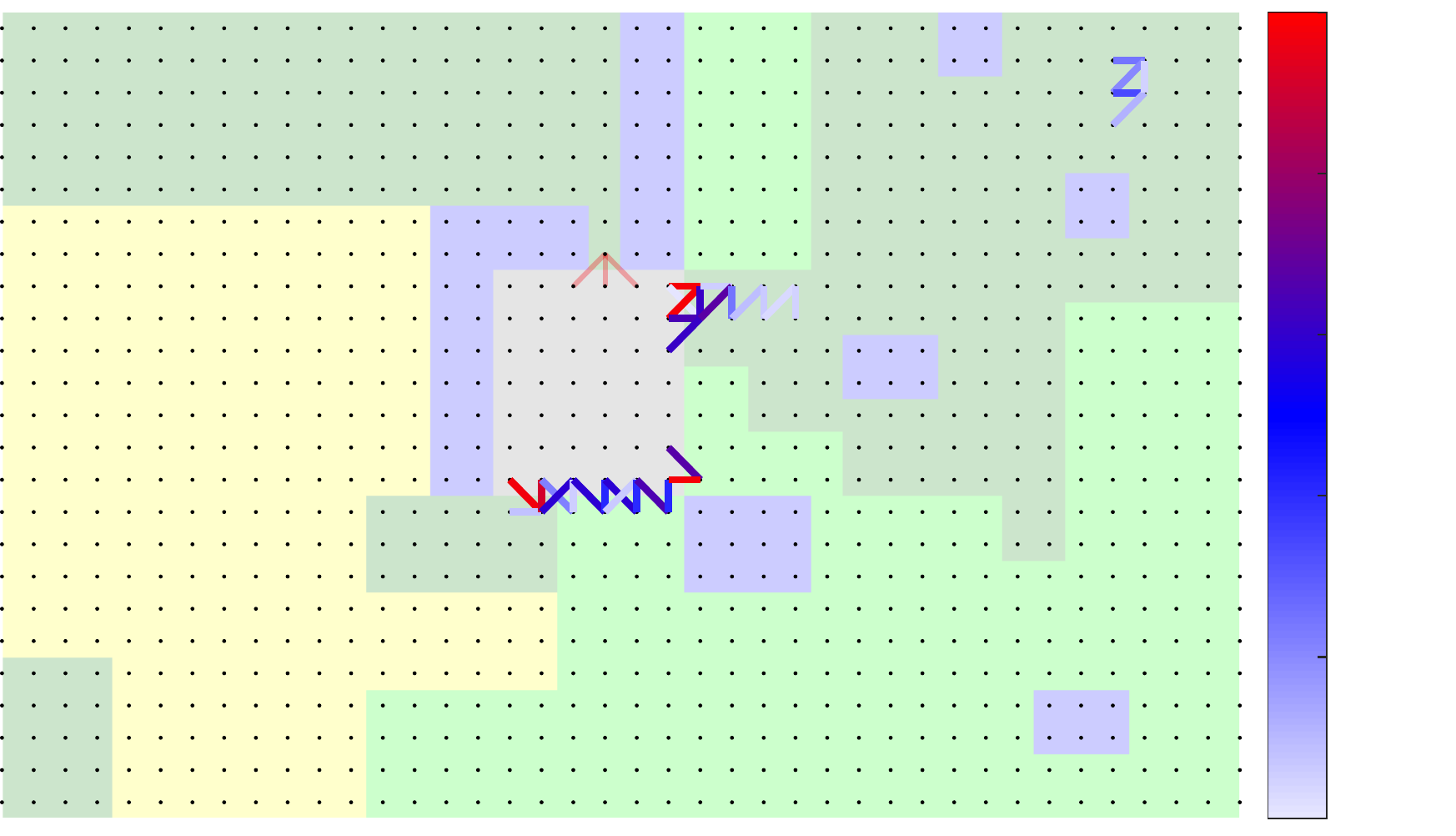}}%
    \put(0.92195404,0.00040657){\makebox(0,0)[lt]{\lineheight{1.25}\smash{\begin{tabular}[t]{l}\footnotesize 0\end{tabular}}}}%
    \put(0.92195404,0.1111276){\makebox(0,0)[lt]{\lineheight{1.25}\smash{\begin{tabular}[t]{l}\footnotesize 0.2\end{tabular}}}}%
    \put(0.92195404,0.22184863){\makebox(0,0)[lt]{\lineheight{1.25}\smash{\begin{tabular}[t]{l}\footnotesize 0.4\end{tabular}}}}%
    \put(0.92195404,0.33256966){\makebox(0,0)[lt]{\lineheight{1.25}\smash{\begin{tabular}[t]{l}\footnotesize 0.6\end{tabular}}}}%
    \put(0.92195404,0.4432907){\makebox(0,0)[lt]{\lineheight{1.25}\smash{\begin{tabular}[t]{l}\footnotesize0.8\end{tabular}}}}%
    \put(0.92195404,0.55401173){\makebox(0,0)[lt]{\lineheight{1.25}\smash{\begin{tabular}[t]{l}\footnotesize1\end{tabular}}}}%
    \put(0.99999,0.11408757){\rotatebox{90}{\makebox(0,0)[lt]{\lineheight{1.25}\smash{\begin{tabular}[t]{l}Resource allocation\end{tabular}}}}}%
    \put(0,0){\includegraphics[width=\unitlength,page=2]{Ex2bN.pdf}}%
  \end{picture}%
\endgroup%
  
      \caption{k=2}
      \label{fig:FG10bA}
  \end{subfigure}
  \begin{subfigure}[b]{0.45\linewidth}
        \def\svgwidth{1\textwidth}
\begingroup%
  \makeatletter%
  \providecommand\color[2][]{%
    \errmessage{(Inkscape) Color is used for the text in Inkscape, but the package 'color.sty' is not loaded}%
    \renewcommand\color[2][]{}%
  }%
  \providecommand\transparent[1]{%
    \errmessage{(Inkscape) Transparency is used (non-zero) for the text in Inkscape, but the package 'transparent.sty' is not loaded}%
    \renewcommand\transparent[1]{}%
  }%
  \providecommand\rotatebox[2]{#2}%
  \newcommand*\fsize{\dimexpr\f@size pt\relax}%
  \newcommand*\lineheight[1]{\fontsize{\fsize}{#1\fsize}\selectfont}%
  \ifx\svgwidth\undefined%
    \setlength{\unitlength}{502.61453247bp}%
    \ifx\svgscale\undefined%
      \relax%
    \else%
      \setlength{\unitlength}{\unitlength * \real{\svgscale}}%
    \fi%
  \else%
    \setlength{\unitlength}{\svgwidth}%
  \fi%
  \global\let\svgwidth\undefined%
  \global\let\svgscale\undefined%
  \makeatother%
  \begin{picture}(1,0.57360556)%
    \lineheight{1}%
    \setlength\tabcolsep{0pt}%
    \put(0,0){\includegraphics[width=\unitlength,page=1]{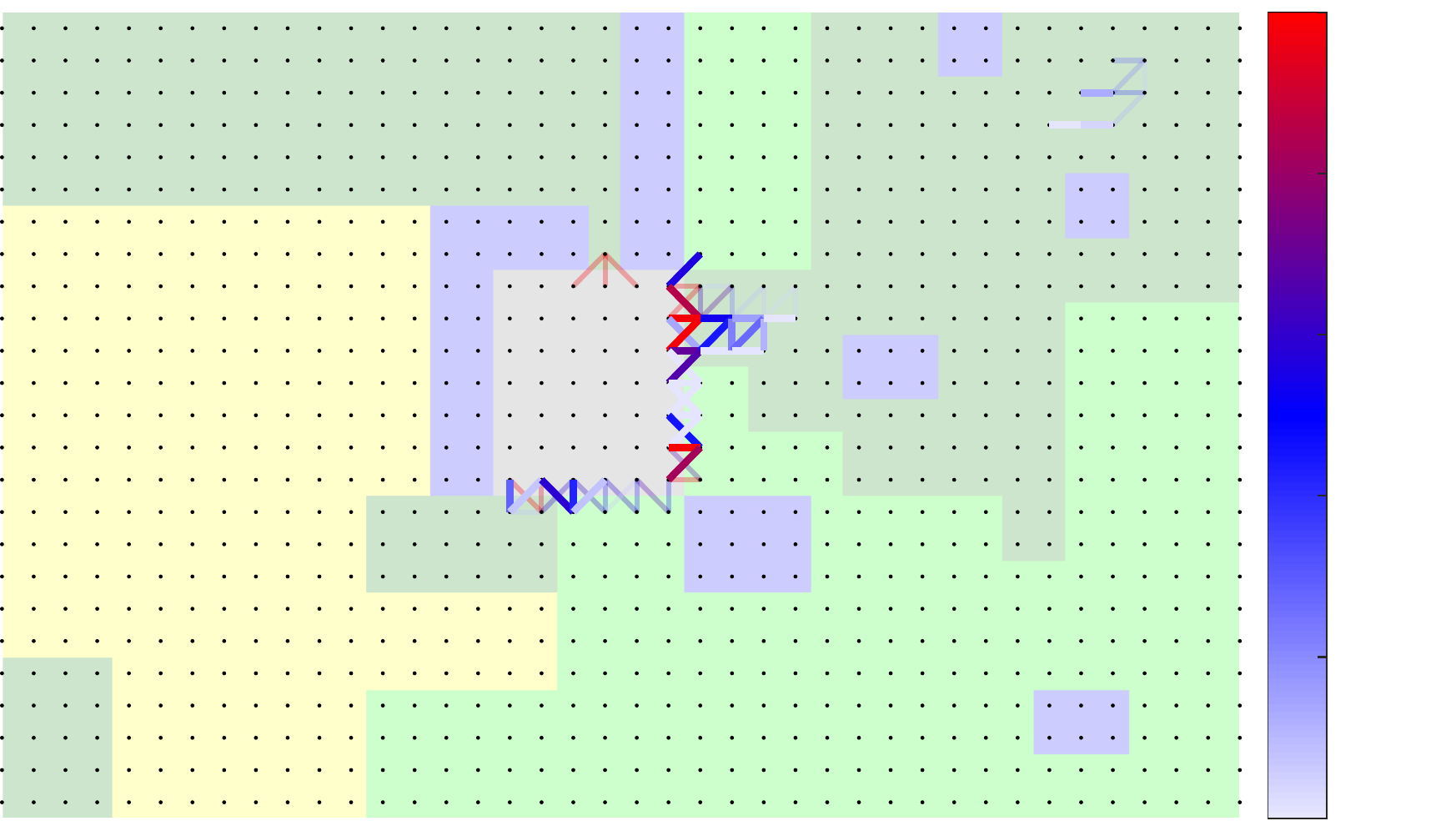}}%
    \put(0.92195404,0.00040657){\makebox(0,0)[lt]{\lineheight{1.25}\smash{\begin{tabular}[t]{l}\footnotesize 0\end{tabular}}}}%
    \put(0.92195404,0.1111276){\makebox(0,0)[lt]{\lineheight{1.25}\smash{\begin{tabular}[t]{l}\footnotesize 0.2\end{tabular}}}}%
    \put(0.92195404,0.22184863){\makebox(0,0)[lt]{\lineheight{1.25}\smash{\begin{tabular}[t]{l}\footnotesize 0.4\end{tabular}}}}%
    \put(0.92195404,0.33256966){\makebox(0,0)[lt]{\lineheight{1.25}\smash{\begin{tabular}[t]{l}\footnotesize 0.6\end{tabular}}}}%
    \put(0.92195404,0.4432907){\makebox(0,0)[lt]{\lineheight{1.25}\smash{\begin{tabular}[t]{l}\footnotesize0.8\end{tabular}}}}%
    \put(0.92195404,0.55401173){\makebox(0,0)[lt]{\lineheight{1.25}\smash{\begin{tabular}[t]{l}\footnotesize1\end{tabular}}}}%
    \put(0.99999,0.11408757){\rotatebox{90}{\makebox(0,0)[lt]{\lineheight{1.25}\smash{\begin{tabular}[t]{l}Resource allocation\end{tabular}}}}}%
    \put(0,0){\includegraphics[width=\unitlength,page=2]{Ex2cN.pdf}}%
  \end{picture}%
\endgroup%
  
      \caption{k=3}
      \label{fig:FG10cA}
  \end{subfigure}
  ~
     \begin{subfigure}[b]{0.45\linewidth}
        \def\svgwidth{1\textwidth}
\begingroup%
  \makeatletter%
  \providecommand\color[2][]{%
    \errmessage{(Inkscape) Color is used for the text in Inkscape, but the package 'color.sty' is not loaded}%
    \renewcommand\color[2][]{}%
  }%
  \providecommand\transparent[1]{%
    \errmessage{(Inkscape) Transparency is used (non-zero) for the text in Inkscape, but the package 'transparent.sty' is not loaded}%
    \renewcommand\transparent[1]{}%
  }%
  \providecommand\rotatebox[2]{#2}%
  \newcommand*\fsize{\dimexpr\f@size pt\relax}%
  \newcommand*\lineheight[1]{\fontsize{\fsize}{#1\fsize}\selectfont}%
  \ifx\svgwidth\undefined%
    \setlength{\unitlength}{502.61453247bp}%
    \ifx\svgscale\undefined%
      \relax%
    \else%
      \setlength{\unitlength}{\unitlength * \real{\svgscale}}%
    \fi%
  \else%
    \setlength{\unitlength}{\svgwidth}%
  \fi%
  \global\let\svgwidth\undefined%
  \global\let\svgscale\undefined%
  \makeatother%
  \begin{picture}(1,0.57360556)%
    \lineheight{1}%
    \setlength\tabcolsep{0pt}%
    \put(0,0){\includegraphics[width=\unitlength,page=1]{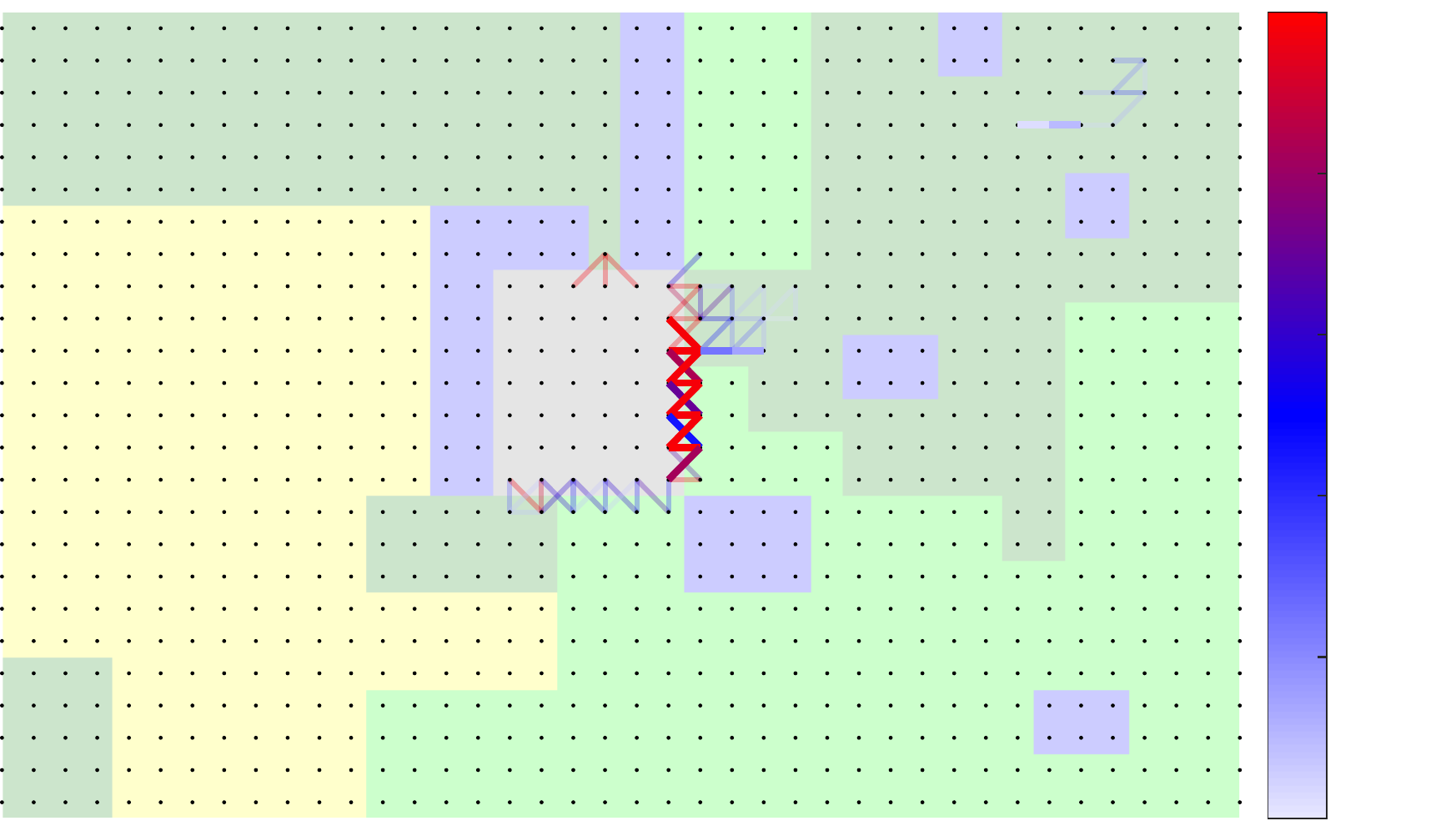}}%
    \put(0.92195404,0.00040657){\makebox(0,0)[lt]{\lineheight{1.25}\smash{\begin{tabular}[t]{l}\footnotesize 0\end{tabular}}}}%
    \put(0.92195404,0.1111276){\makebox(0,0)[lt]{\lineheight{1.25}\smash{\begin{tabular}[t]{l}\footnotesize 0.2\end{tabular}}}}%
    \put(0.92195404,0.22184863){\makebox(0,0)[lt]{\lineheight{1.25}\smash{\begin{tabular}[t]{l}\footnotesize 0.4\end{tabular}}}}%
    \put(0.92195404,0.33256966){\makebox(0,0)[lt]{\lineheight{1.25}\smash{\begin{tabular}[t]{l}\footnotesize 0.6\end{tabular}}}}%
    \put(0.92195404,0.4432907){\makebox(0,0)[lt]{\lineheight{1.25}\smash{\begin{tabular}[t]{l}\footnotesize0.8\end{tabular}}}}%
    \put(0.92195404,0.55401173){\makebox(0,0)[lt]{\lineheight{1.25}\smash{\begin{tabular}[t]{l}\footnotesize1\end{tabular}}}}%
    \put(0.99999,0.11408757){\rotatebox{90}{\makebox(0,0)[lt]{\lineheight{1.25}\smash{\begin{tabular}[t]{l}Resource allocation\end{tabular}}}}}%
    \put(0,0){\includegraphics[width=\unitlength,page=2]{Ex2dN.pdf}}%
  \end{picture}%
\endgroup%
  
\caption{k=K}
      \label{fig:FG10dA}
  \end{subfigure}
      \caption{Resource allocation for a 4-stage approach with $\Gamma^k=10$. Resources shown are total resources up to time $k$ where new resources added have a thicker line width.}     
\label{fig:FG10A}
\end{figure*}

To investigate scalability of the multi-stage approach, we recorded solver time for solving the problem with time steps up to $K=10$ and the results are displayed in Fig. \ref{fig:runK}. The problems were solved with MOSEK v9.2 on an Intel i7, 2.6GHz, 8GB RAM. It can be seen that the MOSEK solver time goes up linearly per time step, which suggests the method is scalable to large number of time steps. Scalability with respect to number of nodes was investigated for a similar (but single-stage) method in \cite{LCSS2021} and also found to be linear.  

It is worth drawing attention to the fact that in this problem there are 3357 edges that can either have resources allocated to them or not. Choosing 10 edges for resource allocation per each time step, over a horizon of 10 time steps, results in a total of approximately $8.7 \times 10^{286}$ possible resource allocation patterns. Despite such a remarkable problem size (in its combinatorial form), our convex formulation finds good solutions in less then a minute on a standard laptop. 



\begin{figure}
\centering
   \def\svgwidth{0.85\linewidth}
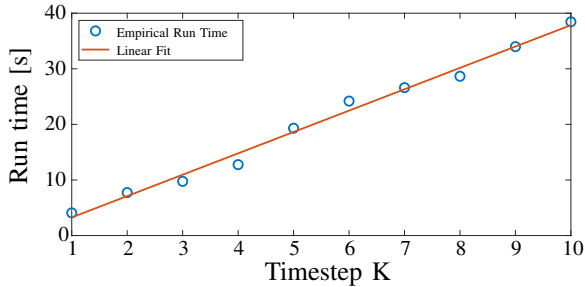
\caption{MOSEK solver time in seconds for solving the wildfire example with time steps up to $K=10$.}
\label{fig:runK}
\end{figure}

Finally, to illustrate the effect of the estimated outbreak probability $\hat x_i^1$, we consider a scenario with a precise knowledge of the outbreak location, as visualized in Fig. \ref{fig:P1LL}. Keeping all other parameters the same, we obtain Fig. \ref{fig:FGL}. It can be seen that in the first time steps $k=1,2$ the outbreak area is aggressively suppressed, before resources are allocated to the path the fire would take to the city. 

\begin{figure}
\centering
   \def\svgwidth{0.85\linewidth}
\begingroup%
  \makeatletter%
  \providecommand\color[2][]{%
    \errmessage{(Inkscape) Color is used for the text in Inkscape, but the package 'color.sty' is not loaded}%
    \renewcommand\color[2][]{}%
  }%
  \providecommand\transparent[1]{%
    \errmessage{(Inkscape) Transparency is used (non-zero) for the text in Inkscape, but the package 'transparent.sty' is not loaded}%
    \renewcommand\transparent[1]{}%
  }%
  \providecommand\rotatebox[2]{#2}%
  \newcommand*\fsize{\dimexpr\f@size pt\relax}%
  \newcommand*\lineheight[1]{\fontsize{\fsize}{#1\fsize}\selectfont}%
  \ifx\svgwidth\undefined%
    \setlength{\unitlength}{512.40431213bp}%
    \ifx\svgscale\undefined%
      \relax%
    \else%
      \setlength{\unitlength}{\unitlength * \real{\svgscale}}%
    \fi%
  \else%
    \setlength{\unitlength}{\svgwidth}%
  \fi%
  \global\let\svgwidth\undefined%
  \global\let\svgscale\undefined%
  \makeatother%
  \begin{picture}(1,0.55093699)%
    \lineheight{1}%
    \setlength\tabcolsep{0pt}%
    \put(0,0){\includegraphics[width=\unitlength,page=1]{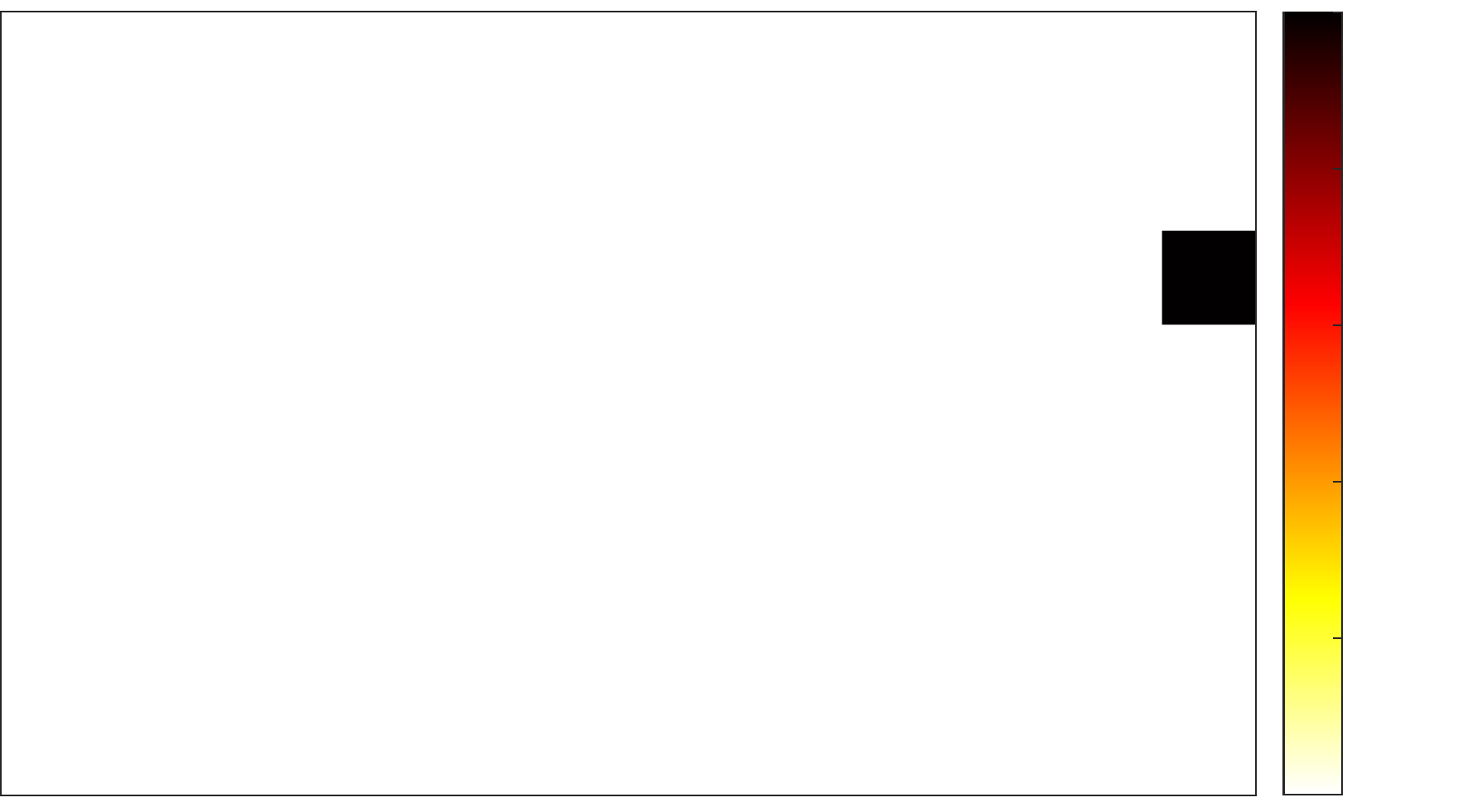}}%
    \put(0.92641682,0.0003988){\makebox(0,0)[lt]{\lineheight{1.25}\smash{\begin{tabular}[t]{l}\scriptsize 0\end{tabular}}}}%
    \put(0.92641682,0.10666254){\makebox(0,0)[lt]{\lineheight{1.25}\smash{\begin{tabular}[t]{l}\scriptsize 0.2\end{tabular}}}}%
    \put(0.92641682,0.21292628){\makebox(0,0)[lt]{\lineheight{1.25}\smash{\begin{tabular}[t]{l}\scriptsize 0.4\end{tabular}}}}%
    \put(0.92641682,0.31919003){\makebox(0,0)[lt]{\lineheight{1.25}\smash{\begin{tabular}[t]{l}\scriptsize 0.6\end{tabular}}}}%
    \put(0.92641682,0.42545377){\makebox(0,0)[lt]{\lineheight{1.25}\smash{\begin{tabular}[t]{l}\scriptsize 0.8\end{tabular}}}}%
    \put(0.92641682,0.53171751){\makebox(0,0)[lt]{\lineheight{1.25}\smash{\begin{tabular}[t]{l}\scriptsize 1\end{tabular}}}}%
    \put(0.99960122,0.09067852){\rotatebox{90}{\makebox(0,0)[lt]{\lineheight{1.25}\smash{\begin{tabular}[t]{l}\small Outbreak Probability\end{tabular}}}}}%
    \put(0,0){\includegraphics[width=\unitlength,page=2]{OmlijnL.pdf}}%
  \end{picture}%
\endgroup%

\caption{Outbreak probability for a scenario with precisely-known initial fire outbreak location.}
\label{fig:P1LL}
\end{figure}

\begin{figure*}
\centering
\begin{subfigure}[b]{0.45\linewidth}
        \def\svgwidth{1\textwidth}
\begingroup%
  \makeatletter%
  \providecommand\color[2][]{%
    \errmessage{(Inkscape) Color is used for the text in Inkscape, but the package 'color.sty' is not loaded}%
    \renewcommand\color[2][]{}%
  }%
  \providecommand\transparent[1]{%
    \errmessage{(Inkscape) Transparency is used (non-zero) for the text in Inkscape, but the package 'transparent.sty' is not loaded}%
    \renewcommand\transparent[1]{}%
  }%
  \providecommand\rotatebox[2]{#2}%
  \newcommand*\fsize{\dimexpr\f@size pt\relax}%
  \newcommand*\lineheight[1]{\fontsize{\fsize}{#1\fsize}\selectfont}%
  \ifx\svgwidth\undefined%
    \setlength{\unitlength}{502.61453247bp}%
    \ifx\svgscale\undefined%
      \relax%
    \else%
      \setlength{\unitlength}{\unitlength * \real{\svgscale}}%
    \fi%
  \else%
    \setlength{\unitlength}{\svgwidth}%
  \fi%
  \global\let\svgwidth\undefined%
  \global\let\svgscale\undefined%
  \makeatother%
  \begin{picture}(1,0.57360556)%
    \lineheight{1}%
    \setlength\tabcolsep{0pt}%
    \put(0,0){\includegraphics[width=\unitlength,page=1]{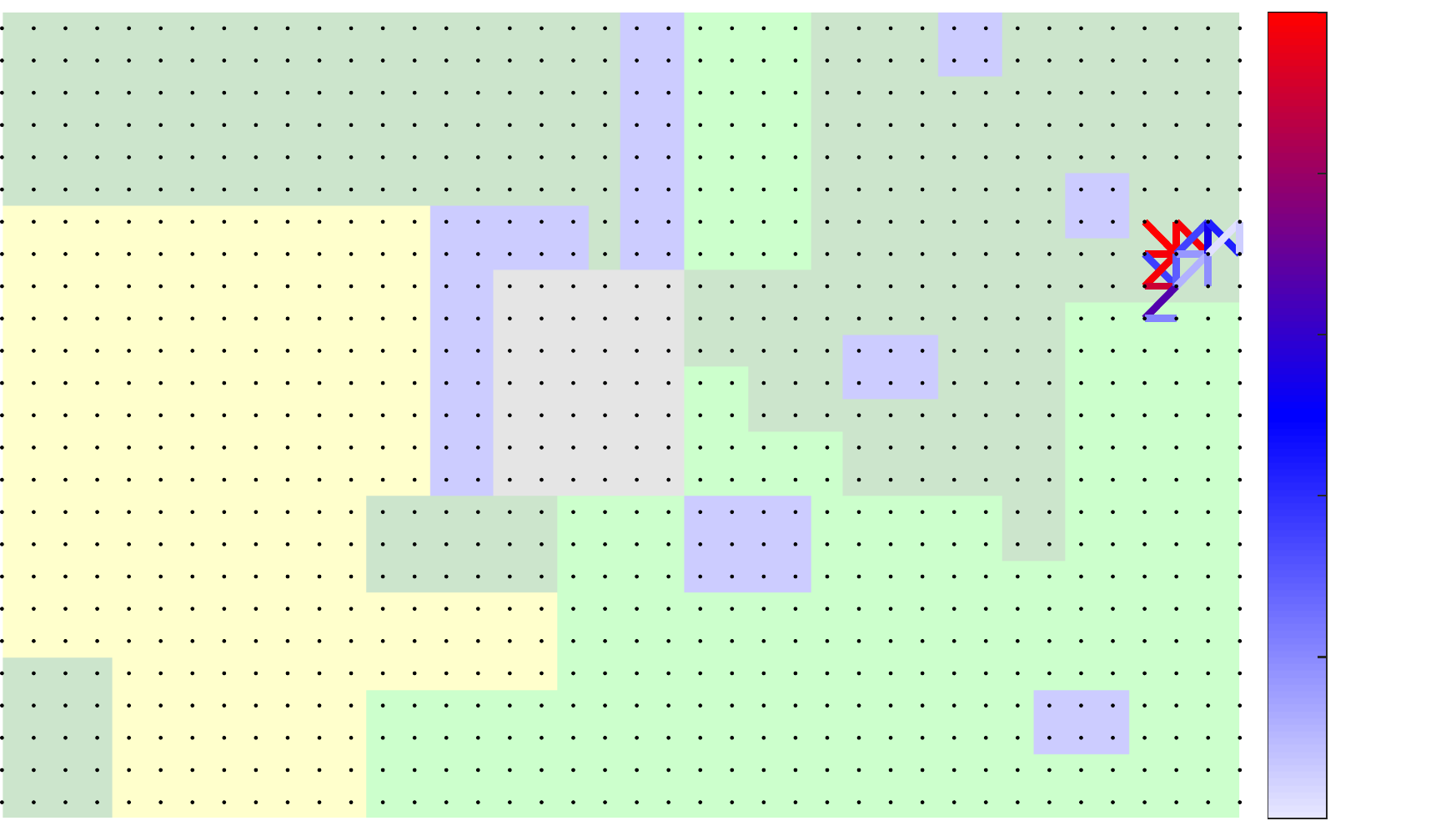}}%
    \put(0.92195404,0.00040657){\makebox(0,0)[lt]{\lineheight{1.25}\smash{\begin{tabular}[t]{l}\footnotesize 0\end{tabular}}}}%
    \put(0.92195404,0.1111276){\makebox(0,0)[lt]{\lineheight{1.25}\smash{\begin{tabular}[t]{l}\footnotesize 0.2\end{tabular}}}}%
    \put(0.92195404,0.22184863){\makebox(0,0)[lt]{\lineheight{1.25}\smash{\begin{tabular}[t]{l}\footnotesize 0.4\end{tabular}}}}%
    \put(0.92195404,0.33256966){\makebox(0,0)[lt]{\lineheight{1.25}\smash{\begin{tabular}[t]{l}\footnotesize 0.6\end{tabular}}}}%
    \put(0.92195404,0.4432907){\makebox(0,0)[lt]{\lineheight{1.25}\smash{\begin{tabular}[t]{l}\footnotesize0.8\end{tabular}}}}%
    \put(0.92195404,0.55401173){\makebox(0,0)[lt]{\lineheight{1.25}\smash{\begin{tabular}[t]{l}\footnotesize1\end{tabular}}}}%
    \put(0.99999,0.11408757){\rotatebox{90}{\makebox(0,0)[lt]{\lineheight{1.25}\smash{\begin{tabular}[t]{l}Resource allocation\end{tabular}}}}}%
    \put(0,0){\includegraphics[width=\unitlength,page=2]{Ex2LaN10.pdf}}%
  \end{picture}%
\endgroup%
  
      \caption{k=1}
      \label{fig:FNAa}
  \end{subfigure}
  ~ 
     \begin{subfigure}[b]{0.45\linewidth}
        \def\svgwidth{1\textwidth}
\begingroup%
  \makeatletter%
  \providecommand\color[2][]{%
    \errmessage{(Inkscape) Color is used for the text in Inkscape, but the package 'color.sty' is not loaded}%
    \renewcommand\color[2][]{}%
  }%
  \providecommand\transparent[1]{%
    \errmessage{(Inkscape) Transparency is used (non-zero) for the text in Inkscape, but the package 'transparent.sty' is not loaded}%
    \renewcommand\transparent[1]{}%
  }%
  \providecommand\rotatebox[2]{#2}%
  \newcommand*\fsize{\dimexpr\f@size pt\relax}%
  \newcommand*\lineheight[1]{\fontsize{\fsize}{#1\fsize}\selectfont}%
  \ifx\svgwidth\undefined%
    \setlength{\unitlength}{502.61453247bp}%
    \ifx\svgscale\undefined%
      \relax%
    \else%
      \setlength{\unitlength}{\unitlength * \real{\svgscale}}%
    \fi%
  \else%
    \setlength{\unitlength}{\svgwidth}%
  \fi%
  \global\let\svgwidth\undefined%
  \global\let\svgscale\undefined%
  \makeatother%
  \begin{picture}(1,0.57360556)%
    \lineheight{1}%
    \setlength\tabcolsep{0pt}%
    \put(0,0){\includegraphics[width=\unitlength,page=1]{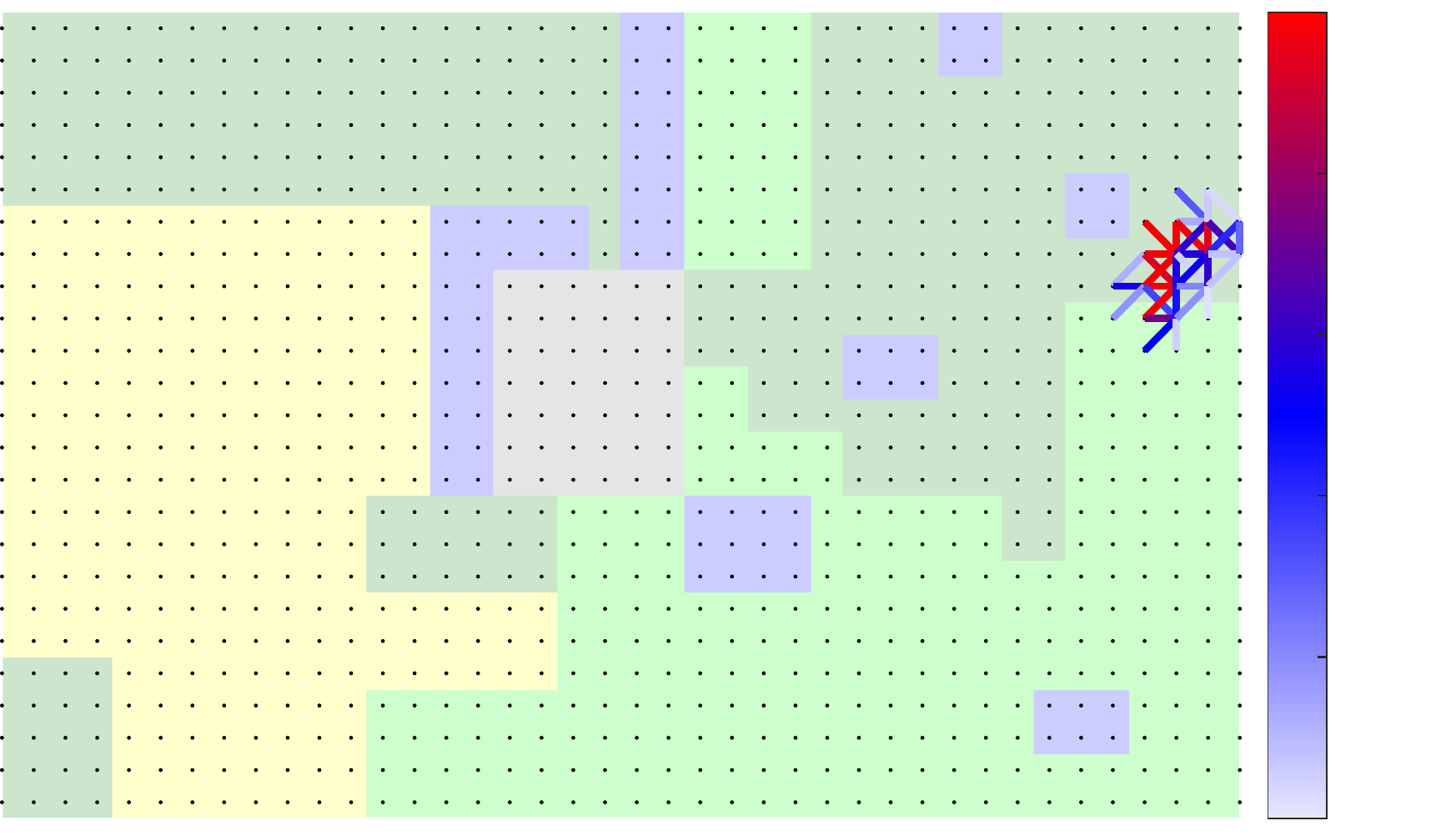}}%
    \put(0.92195404,0.00040657){\makebox(0,0)[lt]{\lineheight{1.25}\smash{\begin{tabular}[t]{l}\footnotesize 0\end{tabular}}}}%
    \put(0.92195404,0.1111276){\makebox(0,0)[lt]{\lineheight{1.25}\smash{\begin{tabular}[t]{l}\footnotesize 0.2\end{tabular}}}}%
    \put(0.92195404,0.22184863){\makebox(0,0)[lt]{\lineheight{1.25}\smash{\begin{tabular}[t]{l}\footnotesize 0.4\end{tabular}}}}%
    \put(0.92195404,0.33256966){\makebox(0,0)[lt]{\lineheight{1.25}\smash{\begin{tabular}[t]{l}\footnotesize 0.6\end{tabular}}}}%
    \put(0.92195404,0.4432907){\makebox(0,0)[lt]{\lineheight{1.25}\smash{\begin{tabular}[t]{l}\footnotesize0.8\end{tabular}}}}%
    \put(0.92195404,0.55401173){\makebox(0,0)[lt]{\lineheight{1.25}\smash{\begin{tabular}[t]{l}\footnotesize1\end{tabular}}}}%
    \put(0.99999,0.11408757){\rotatebox{90}{\makebox(0,0)[lt]{\lineheight{1.25}\smash{\begin{tabular}[t]{l}Resource allocation\end{tabular}}}}}%
    \put(0,0){\includegraphics[width=\unitlength,page=2]{Ex2LbN10.pdf}}%
  \end{picture}%
\endgroup%
 
\caption{k=2}
      \label{fig:FNAb}
  \end{subfigure}
  \begin{subfigure}[b]{0.45\linewidth}
        \def\svgwidth{1\textwidth}
\begingroup%
  \makeatletter%
  \providecommand\color[2][]{%
    \errmessage{(Inkscape) Color is used for the text in Inkscape, but the package 'color.sty' is not loaded}%
    \renewcommand\color[2][]{}%
  }%
  \providecommand\transparent[1]{%
    \errmessage{(Inkscape) Transparency is used (non-zero) for the text in Inkscape, but the package 'transparent.sty' is not loaded}%
    \renewcommand\transparent[1]{}%
  }%
  \providecommand\rotatebox[2]{#2}%
  \newcommand*\fsize{\dimexpr\f@size pt\relax}%
  \newcommand*\lineheight[1]{\fontsize{\fsize}{#1\fsize}\selectfont}%
  \ifx\svgwidth\undefined%
    \setlength{\unitlength}{502.61453247bp}%
    \ifx\svgscale\undefined%
      \relax%
    \else%
      \setlength{\unitlength}{\unitlength * \real{\svgscale}}%
    \fi%
  \else%
    \setlength{\unitlength}{\svgwidth}%
  \fi%
  \global\let\svgwidth\undefined%
  \global\let\svgscale\undefined%
  \makeatother%
  \begin{picture}(1,0.57360556)%
    \lineheight{1}%
    \setlength\tabcolsep{0pt}%
    \put(0,0){\includegraphics[width=\unitlength,page=1]{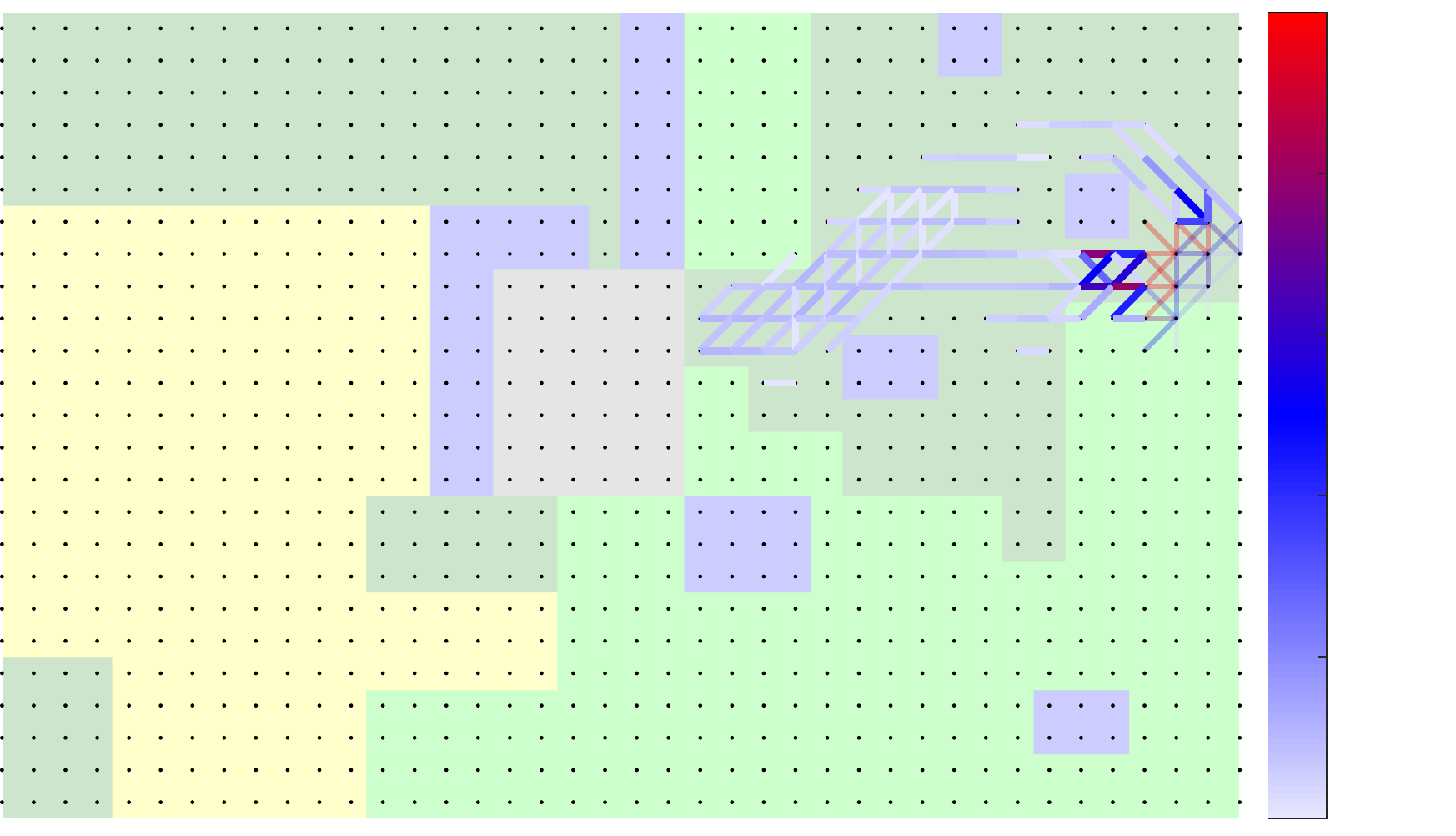}}%
    \put(0.92195404,0.00040657){\makebox(0,0)[lt]{\lineheight{1.25}\smash{\begin{tabular}[t]{l}\footnotesize 0\end{tabular}}}}%
    \put(0.92195404,0.1111276){\makebox(0,0)[lt]{\lineheight{1.25}\smash{\begin{tabular}[t]{l}\footnotesize 0.2\end{tabular}}}}%
    \put(0.92195404,0.22184863){\makebox(0,0)[lt]{\lineheight{1.25}\smash{\begin{tabular}[t]{l}\footnotesize 0.4\end{tabular}}}}%
    \put(0.92195404,0.33256966){\makebox(0,0)[lt]{\lineheight{1.25}\smash{\begin{tabular}[t]{l}\footnotesize 0.6\end{tabular}}}}%
    \put(0.92195404,0.4432907){\makebox(0,0)[lt]{\lineheight{1.25}\smash{\begin{tabular}[t]{l}\footnotesize0.8\end{tabular}}}}%
    \put(0.92195404,0.55401173){\makebox(0,0)[lt]{\lineheight{1.25}\smash{\begin{tabular}[t]{l}\footnotesize1\end{tabular}}}}%
    \put(0.99999,0.11408757){\rotatebox{90}{\makebox(0,0)[lt]{\lineheight{1.25}\smash{\begin{tabular}[t]{l}Resource allocation\end{tabular}}}}}%
    \put(0,0){\includegraphics[width=\unitlength,page=2]{Ex2LcN10.pdf}}%
  \end{picture}%
\endgroup%
 
      \caption{k=3}
      \label{fig:FNAc}
  \end{subfigure}
  ~ 
     \begin{subfigure}[b]{0.45\linewidth}
        \def\svgwidth{1\textwidth}
\begingroup%
  \makeatletter%
  \providecommand\color[2][]{%
    \errmessage{(Inkscape) Color is used for the text in Inkscape, but the package 'color.sty' is not loaded}%
    \renewcommand\color[2][]{}%
  }%
  \providecommand\transparent[1]{%
    \errmessage{(Inkscape) Transparency is used (non-zero) for the text in Inkscape, but the package 'transparent.sty' is not loaded}%
    \renewcommand\transparent[1]{}%
  }%
  \providecommand\rotatebox[2]{#2}%
  \newcommand*\fsize{\dimexpr\f@size pt\relax}%
  \newcommand*\lineheight[1]{\fontsize{\fsize}{#1\fsize}\selectfont}%
  \ifx\svgwidth\undefined%
    \setlength{\unitlength}{502.61453247bp}%
    \ifx\svgscale\undefined%
      \relax%
    \else%
      \setlength{\unitlength}{\unitlength * \real{\svgscale}}%
    \fi%
  \else%
    \setlength{\unitlength}{\svgwidth}%
  \fi%
  \global\let\svgwidth\undefined%
  \global\let\svgscale\undefined%
  \makeatother%
  \begin{picture}(1,0.57360556)%
    \lineheight{1}%
    \setlength\tabcolsep{0pt}%
    \put(0,0){\includegraphics[width=\unitlength,page=1]{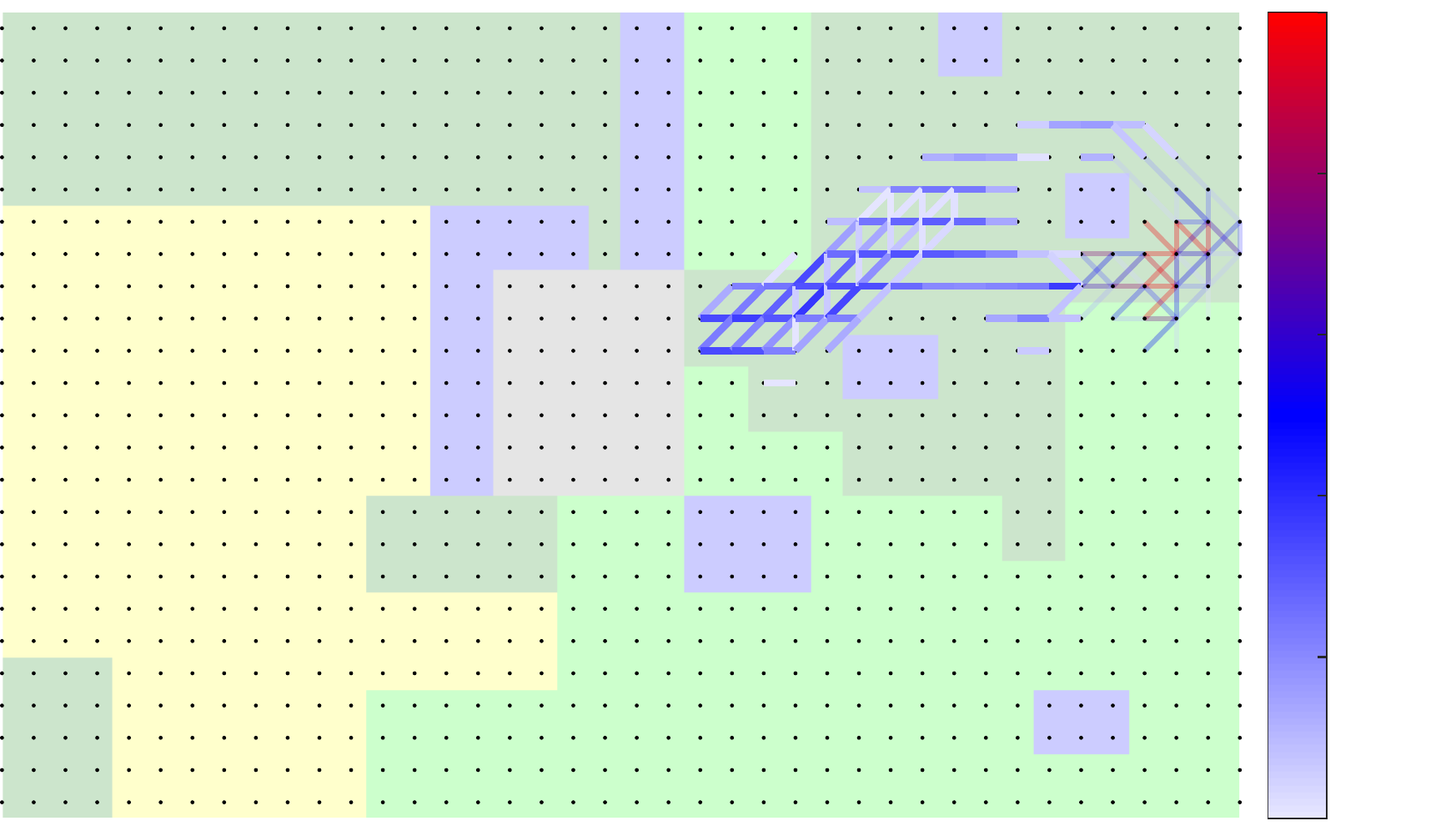}}%
    \put(0.92195404,0.00040657){\makebox(0,0)[lt]{\lineheight{1.25}\smash{\begin{tabular}[t]{l}\footnotesize 0\end{tabular}}}}%
    \put(0.92195404,0.1111276){\makebox(0,0)[lt]{\lineheight{1.25}\smash{\begin{tabular}[t]{l}\footnotesize 0.2\end{tabular}}}}%
    \put(0.92195404,0.22184863){\makebox(0,0)[lt]{\lineheight{1.25}\smash{\begin{tabular}[t]{l}\footnotesize 0.4\end{tabular}}}}%
    \put(0.92195404,0.33256966){\makebox(0,0)[lt]{\lineheight{1.25}\smash{\begin{tabular}[t]{l}\footnotesize 0.6\end{tabular}}}}%
    \put(0.92195404,0.4432907){\makebox(0,0)[lt]{\lineheight{1.25}\smash{\begin{tabular}[t]{l}\footnotesize0.8\end{tabular}}}}%
    \put(0.92195404,0.55401173){\makebox(0,0)[lt]{\lineheight{1.25}\smash{\begin{tabular}[t]{l}\footnotesize1\end{tabular}}}}%
    \put(0.99999,0.11408757){\rotatebox{90}{\makebox(0,0)[lt]{\lineheight{1.25}\smash{\begin{tabular}[t]{l}Resource allocation\end{tabular}}}}}%
    \put(0,0){\includegraphics[width=\unitlength,page=2]{Ex2LdN10.pdf}}%
  \end{picture}%
\endgroup%
 
\caption{k=K}
      \label{fig:FNAd}
  \end{subfigure}
      \caption{Resource allocation for a 4-stage approach with precisely-known initial fire outbreak location (Fig. \ref{fig:P1LL}) with $\Gamma^k=10$. Resources shown are total resources up to time $k$ where new resources added have a thicker line width.}     
\label{fig:FGL}
\end{figure*}

\section{CONCLUSIONS}
In this paper we presented a multi-stage approach to minimize the risk of spreading processes over networks by use of sparse control. A convex exponential programming formulation was presented including sparsity inducing logarithmic resource models. With both epidemic and wildfire examples the use and potential of the proposed method was conveyed. Future work will include more realistic propagation models and time-varying networks. 




%
%
%


\bibliographystyle{IEEEtran}
\bibliography{ACFR2021}

\end{document}